%% file: main.tex
\documentclass{llncs}
\usepackage{multicol}
\usepackage{subfigure}
\usepackage{pstricks,pst-grad}
\usepackage{amssymb}
\usepackage{stmaryrd}
\usepackage{fix-cm}
\usepackage{amsmath}
\usepackage{color}
\usepackage{float}
\usepackage{wrapfig}
\usepackage{wasysym}
\usepackage{tikz}
\usetikzlibrary{arrows,shapes,snakes}
\usepackage{multirow}
\usepackage{subfigure}
\usepackage{float}
\usepackage{verbatim}
\usepackage[ruled,vlined]{algorithm2e}
\usepackage{pgf}
\usepackage{tikz}
\usetikzlibrary{shapes,snakes}
\usetikzlibrary{positioning}
\usepackage[utf8]{inputenc}
\usetikzlibrary{arrows,automata}
\usetikzlibrary{positioning}
\usetikzlibrary{arrows,shapes,fit,backgrounds}

\usepackage[normalem]{ulem}

\usepackage{ifthen}
\newcommand{\techRep}{true} 
\newcommand{\iftechrep}{\ifthenelse{\equal{\techRep}{true}}}

\input{macros}

\begin{document}

\author{Javier Esparza \and Andreas Gaiser} 
\title{Probabilistic Abstractions with Arbitrary Domains}
\institute{Fakult{\"a}t f{\"u}r Informatik, Technische Universit{\"a}t M{\"u}nchen, Germany \\ \texttt{$\{$esparza,gaiser$\}$@model.in.tum.de}}  
\maketitle

\iftechrep{\andreas{Technical Report-Version}}{\andreas{Konferenz-Paper}}

\pagestyle{headings}  
\LinesNotNumbered

\begin{abstract}

Recent work by Hermanns et al. and Kattenbelt et al. has extended
counterexample-guided abstraction refinement (CEGAR) to probabilistic
programs. 
These approaches are limited to predicate abstraction.
We present a novel technique, based on the abstract reachability tree
recently introduced by Gulavani et al., that can use arbitrary abstract
domains and widening operators (in the sense of Abstract Interpretation).
We show how suitable widening operators can deduce loop invariants
difficult to find for predicate abstraction, and propose refinement techniques.

\end{abstract}

\input{intro_game.tex}

\input{abstraction_game-new.tex}

\input{refinements.tex}

\input{experiments.tex}

\input{conclusion.tex}

\subsubsection*{Acknowledgements.}
We thank Holger Hermanns, Ernst-Moritz Hahn, and Luis M.F. Fioriti for 
valuable comments, Bj{\"o}rn Wachter for many
discussions during the second author's stay at the University of Oxford,
made possible by Joel Ouaknine, and five anonymous reviewers for 
helpful remarks. The second author is supported by the DFG Graduiertenkolleg 1480 (PUMA).


%

\bibliographystyle{plain} 
\bibliography{literatur}

\iftechrep{\input{appendix.tex}}{}
\end{document}

%% file: macros.tex
\newcommand{\sM}{\mathcal{M}}

\newcommand{\sV}{\mathcal{V}}

\newcommand{\node}{\text{node}}

\newcommand{\rng}{\text{rng}}

\newcommand{\trans}{\text{Trans}_\sV}
\newcommand{\fopt}{\text{fopt}}
\newcommand{\configs}{\Sigma_\sV}
\newcommand{\pred}{\text{pred}}
\newcommand{\worklist}{\text{work}}
\newcommand{\sem}[1]{\llbracket #1 \rrbracket}
\newcommand{\abs}[1]{| #1 |}
\newcommand{\cyl}[1]{\text{Cyl}(#1)}
\newcommand{\reach}[1]{\text{Reach}(#1)}
\newcommand{\tup}[1]{\langle #1 \rangle}
\newcommand{\sG}{\mathcal{G}}
\newcommand{\sC}{\mathcal{C}}
\newcommand{\val}[2]{\text{val}^{\,#2}(#1)}
\newcommand{\dist}{\text{Dist}}

\newcommand{\mpphi}{\mathcal{M}[\phi]}
\newcommand{\gpp}{\mathcal{G}[\phi_1, \phi_2]}

\definecolor{orange3}{rgb}{1.0,0.2538,0.1681}
\definecolor{blau}{rgb}{0,.39608,0.74118}
\definecolor{rot}{rgb}{0.79216,.12941,0.24706}
\definecolor{dgruen}{rgb}{0,.7,0}
\newcommand\Color[2]{\textcolor{#1}{#2}}
\newcommand{\nb}[1]{\marginpar{\scriptsize #1}}

\newcommand\Andreas[1]{\Color{rot}{(*#1*)}}
\newcommand\Javier[1]{\Color{dgruen}{(*#1*)}}
\newcommand\javier[1]{\nb{\Color{dgruen}{(*#1*)}}}
\newcommand\andreas[1]{\nb{\Color{rot}{(*#1*)}}}

 \renewcommand\Andreas[1]{}
 \renewcommand\Javier[1]{}
 \renewcommand\javier[1]{}
 \renewcommand\andreas[1]{}

%% file: intro_game.tex
\section{Introduction}
Abstraction techniques are crucial for the automatic verification of systems
with a finite but very large or infinite state space. The
Abstract Interpretation framework provides the mathematical basis of
abstraction \cite{cousot}. 
Recent work has extended abstraction techniques to probabilistic systems 
using games \cite{wachter,hermanns,kattenbelt,bjoern-phd,probtrans}. 
The systems (e.g. probabilistic programs) are given semantics in terms of Markov
Decision Processes (MDPs), which can model nondeterminism and (using interleaving semantics)
concurrency. The key idea is to abstract the MDP into a stochastic 2-Player game, distinguishing between nondeterminism inherent to the system 
(modeled by the choices of Player 1)
and nondeterminism introduced by the abstraction (modeled by Player 2). 
The construction ensures that the probability of reaching a goal state in the MDP using an
optimal strategy is bounded from above and from below by the
supremum and infimum of the probabilities of reaching the goal in the 2-Player game when Player
1 plays according to an optimal strategy (and Player 2 is free to play in any way)\footnote{In~\cite{kattenbelt}, the roles
of the players are interchanged.}.
An analogous result holds for the probability of reaching a goal state in the MDP using a
pessimal strategy.

The abstraction technique of \cite{kattenbelt,probtrans} and the related 
\cite{wachter,hermanns,bjoern-phd} relies on predicate abstraction:
an abstract state is an equivalence class of concrete states, where 
two concrete states are equivalent if they satisfy the same subset of a given
set of predicates. The concretization of two distinct abstract states
is always disjoint (the {\em disjointness property}). 
If the upper and lower bounds obtained using a set of predicates
are not close enough, the abstraction is
 refined by adding new predicates with the help of
interpolation, analogously to the well-known CEGAR approach for
non-probabilistic systems. 

While predicate abstraction 
has proved very successful, it is known to have a number of shortcomings:
potentially expensive equality and inclusion checks for abstract states, and ``predicate explosion''.
In the non-probabilistic case, the work of Gulavani {\it et al.} has 
extended the CEGAR approach to a broader range of abstract domains \cite{gulavani_old}, in which 
widening operations can be combined with interpolation methods, leading to more efficient 
abstraction algorithms.
We show that the ideas of Gulavani {\it et al.} can also be applied to probabilistic systems, which extends the approaches of 
\cite{kattenbelt,mdpreport,probtrans} to 
arbitrary abstract domains. Given a probabilistic program, an abstract domain and a widening for this domain, we show how to construct an abstract stochastic 2-Player reachability game. 
The disjointness property is not required. We prove that bounds on the 
probability of reaching a goal state in the MDP
can be computed as in \cite{kattenbelt,probtrans}. The proofs 
of \cite{probtrans} use the disjointness property to easily define a
Galois connection between the sets of functions assigning values to the abstract and the concrete states. 
Since there seems to be no easy way to adapt them or the ones from~\cite{kattenbelt,mdpreport} to our construction, 
we show the soundness of our approach
by a new proof that uses different techniques.

We also propose an abstraction refinement technique that adapts the idea of delaying the application
of widenings \cite{astree} to the probabilistic case. The technique delays widenings at the nodes
which are likely to have a larger impact in improving the bounds. 
We present experimental results on several examples.

The paper is organized as follows. In the rest of the introduction we 
discuss related work and informally present the key ideas of our approach by means of examples. 
Section~\ref{sec:preliminaries}
contains preliminaries. Section~\ref{sec:abstraction} formally introduces the abstraction technique, and
proves that games we are considering indeed give us upper resp.
lower bound of the exact minimal and maximal reachability probabilities of reaching a set of states. 
Section~\ref{sec:refinements} shows methods of refining our abstractions
and discusses some experiments.
\paragraph{Related work.}

Besides~\cite{wachter,hermanns,kattenbelt,bjoern-phd,probtrans}, 
Monniaux has studied in~\cite{monniaux1} how to abstract 
probability distributions over program states (instead of the states themselves), but only considers
upper bounds for probabilities, as already pointed out in~\cite{probtrans}. 
In~\cite{monniaux2}, Monniaux analyses different quantitative properties of 
Markov Decision processes, again using abstractions of probability distributions. 
In contrast, our approach constructs an abstraction using 
``non-probabilistic'' domains and widenings and then performs the computation of
strategies and strategy values, which might be used for a refinement of the abstraction.
Finally, in~\cite{hankin} Hankin, Di Pierro, and Wiklicky develop a framework for 
probabilistic Abstract Interpretation which, loosely speaking,
 replaces abstract domains by linear spaces and Galois connections by special linear maps,
and aims at computing expected values of random variables.
In contrast, we stick to the standard framework, since in particular we wish to
apply existing tools, and aim for upper and lower bounds of probabilities.
\subsection{An example}

\label{subsec:example}
Consider the following program, written in pseudo code:

\begin{center}
\begin{minipage}{10cm}
\begin{verbatim}
   int nrp = 0;
1: while (nrp < 100) 
2:    if (rec_pack()) then nrp = nrp+1 else break 
3: if (nrp < 1) then (fail or goto 1) else exit 
\end{verbatim}
\end{minipage}
\end{center}

\noindent where the choice between \verb+fail+ and \verb+goto 1+ is decided by
the environment. The goal of the program is to receive up to $100$ packets
through a network connection. Whenever \texttt{rec\_pack()}
is executed, the connection can break down with probability 0.01, in which case
\texttt{rec\_pack()} returns false and no further packets can be received.
If at least one packet is received (line 3) the program terminates\footnote{It would be more realistic to set another bound, like
20 packets, but with one packet the probabilities are easy to compute.}; otherwise,
the program tries to repair the connection, which may fail or succeed, 
in which case
the program is started again. The choice 
between success and failure is nondeterministic.

We formalize the pseudo code as a \emph{Nondeterministic Probabilistic Program}, abbreviated NPP\footnote{NPPs 
roughly correspond to a subset of the input language of the model checker PRISM~\cite{qproversite}.} (see Fig.~\ref{fig:ex1_code}).
A NPP is a collection of guarded commands. A guarded command consists of a name (e.g. \verb+A1+), 
followed by a guard (e.g. \texttt{(ctr = 1) \& (nrp < 100)})  and a sequence of pairs of probabilities and 
update commands (e.g. \texttt{0.99: (nrp' = nrp+1)}), separated by '\texttt{+}'. 
A \texttt{reach}-line at the end of the NPP 
describes the set of states for which we want to compute the reachability 
probability. We call them the {\em final} states. In our example, reaching \texttt{fail} corresponds to satisfying 
\verb+(ctr = 3) && (nrp < 1)+.
A program execution starts with the initial configuration given by the variable declarations.
The program chooses a guarded command whose guard is enabled (i.e., satisfied) by the current state of the program, 
and selects one of its update commands at random, according to the annotated probabilities. After performing the update, the process 
is repeated until the program reaches a final state.
\begin{figure}
 \centering
\begin{minipage}{4.5cm}
\begin{verbatim}
  int nrp = 0, ctr = 1;
A1: (ctr = 1) & (nrp < 100) 
     -> 0.99:(nrp' = nrp+1) 
      + 0.01:(ctr' = 2);
A2: (ctr = 1) & (nrp >= 100) 
     -> 1:(ctr' = 3);
A3: (ctr = 2) & (nrp < 1)
     -> 1:(ctr' = 1);
A4: (ctr = 2)
     -> 1:(ctr' = 3);
A5: (ctr = 3) & (nrp >= 1)
     -> 1:(ctr' = 3);
reach: (ctr = 3) & (nrp < 1)
\end{verbatim}

\end{minipage}
\begin{minipage}{7.5cm}
\centering
 \tikzset{    
   dim/.style={
           rounded corners,
           minimum height=2em,
           inner sep=2pt,
           text centered,
           },
    state/.style={
           rectangle,
           rounded corners,
           draw=black, very thick,
           minimum height=2em,
           inner sep=2pt,
           text centered,
           },
   act/.style={
           circle,
           draw=black, very thick,
           minimum height=2em,
           inner sep=2pt,
           text centered,
           },
  fail/.style={circle, double, draw=black, inner sep=2pt, minimum height=2em, very thick},
  nothing/.style={
           rectangle,
           rounded corners,
           minimum height=2em,
           inner sep=0pt,
           text centered,
           }
  }
\scalebox{0.55}{
\begin{tikzpicture}[->,>=stealth',scale=1,font=\Large]
 \node[state, anchor=center] (S1) 
 {1,0};
  
 \node[dim] (START) [above=of S1]
 {};

 \node[act] (S2) [below =of S1]
 {\texttt{A1}};
 
 \node[state] (S3) [below left=of S2]
  {1,1};

\node[state] (S5) [below right =of S2]
  {2,0};

\node[act] (S6a) [below =of S5]
  {\texttt{A4}};

\node[act] (S6b) [below right =of S5]
  {\texttt{A3}};

\node[act] (S7) [below =of S3]
  {\texttt{A1}};

\node[act] (S8a) [below left =of S7]
  {1,2};

\node[act] (S8b) [below right=of S7]
  {2,1};

\node[dim] (S9a) [below =of S8a]
  {$\Huge{\ldots}$};

\node[dim] (S9b) [below =of S8b]
  {$\Huge{\ldots}$};

\node[fail] (S10) [below =of S6a]
  {3,0};
 \path[very thick]
       (S1) edge[right ] node [right] {} (S2)
       (START) edge[right ] node [right] {} (S1)
       (S2) edge[left]  node{$0.99$} (S3)
       (S2) edge[right]  node{\hspace{0.1cm}$0.01$} (S5)
       (S5) edge[left]  node{} (S6a)
       (S5) edge[left]  node{} (S6b)
       (S6b) edge[bend right]  node [right] {$1$} (S1)
       (S3) edge  node [right] {} (S7)
       (S7) edge  node [left] {$0.99$} (S8a)
       (S7) edge  node [right] {\hspace{0.1cm}$0.01$} (S8b)
       (S8a) edge  node [right] {} (S9a)
       (S8b) edge  node [right] {} (S9b)
       (S6a) edge  node [right] {\hspace{0.2cm}$1$} (S10);
\end{tikzpicture}
}
\end{minipage}
\caption{Example program and prefix of the corresponding Markov Decision Process.
Actions are drawn as circles, program states $\tup{ctr, nrp}$ as rectangles. 
$\tup{3,0}$ is a final state.}
\label{fig:ex1_code}
\end{figure}


\begin{figure}[htbp]
\hfill
\centering
 \tikzset{
    dim/.style={
           rounded corners,
           minimum height=2em,
           inner sep=2pt,
           text centered,
           },
    p1/.style={
           rectangle,
           rounded corners,
           draw=black, very thick,
           minimum height=2em,
           inner sep=2pt,
           text centered,
           },
    prb/.style=
	  {
           rectangle split,
           rectangle,
	   fill = lightgray,
	   rounded corners,
           draw=black, very thick,
           minimum height=2em,
           inner sep=2pt,
           text centered,
           },
     p2/.style=
	  {
           circle,
	   text width=1.3cm,
           draw=black, very thick,
           minimum height=2em,
           inner sep=2pt,
           text centered,
            },
   deadend/.style={circle,draw=white, very thick,
           minimum height=2em,},
  accept/.style={circle,draw=white, very thick,
           minimum height=2em,},
  }
\scalebox{0.34}{
\begin{tikzpicture}[->,>=stealth',scale=1, font=\Huge]

 \node[p1, anchor=north east] (AS1) 
 {${1, [0,0]}$};

 \node[dim] (ASTART) [above=of AS1]
 {};

 \node[p2] (AS2) [below =of AS1]
 {$\texttt{A1}$};
 
 \node[prb] (AS3) [below =of AS2]
 {${1, [0,0]}$};

 \node[p1] (AS4) [below left =of AS3] 
 {${2, [0,0]}$};

 \node[p2] (AS5) [below left =1.75cm of AS4]
 {\texttt{A3}};

 \node[p2] (AS6) [below =of AS4]
 {\texttt{A4}};

 \node[prb] (AS7) [below =of AS5]
 {${2, [0,0]}$};

 \node[prb] (AS8) [below =1cm of AS6]
 {${2, [0,0]}$};

 \node[p1] (AS9) [below right =of AS3] 
 {${1, [0,\infty)}$};

 \node[p2] (AS10) [below right =1.75cm of AS9] 
 {$\texttt{A1}$};

 \node[p2] (AS11) [below left =1.75cm of AS9] 
 {$\texttt{A2}$};

 \node[prb] (AS12) [below =of AS11]
 {${1, [100,\infty)}$};

 \node[prb] (AS13) [below  =of AS10] 
 {${1, [0, 99]}$};

 \node[p1] (AS15) [below =of AS13] 
 {${2, [0,99]}$};

 \node[p2] (AS16) [below right =1.75cm of AS15] 
 {\texttt{A3}};

  \node[p2] (AS17) [below  =of AS15] 
 {\texttt{A4}};

 \node[prb] (AS18) [below =of AS16] 
 {${2, [0,0]}$};


 \node[prb] (AS20) [below =of AS17] 
 {${2, [0,99]}$};

 \node[p1] (AS21) [below =of AS20] 
 {${3, [0,99]}$};

 \node[p2] (AS21a) [below left =of AS21] 
 {$\texttt{A5}$};

 \node[prb] (AS21b) [below =of AS21a] 
 {$3, [1,99]$};

 \node[p2] (AS22) [below =5cm of AS21] 
 {$\varocircle?$};

 \node[p1] (AS23) [below =of AS12] 
 {${3, [100,\infty)}$};

 \node[p2] (AS23a) [below =of AS23] 
 {\texttt{A5}};

 \node[prb] (AS23b) [below =of AS23a] 
 {${3, [100, \infty)}$};

 \node[p1] (AS24) [below =of AS8] 
 {${3, [0,0]}$};

 \node[p2] (AS24a) [below =of AS24] 
 {$\varocircle?$};

\node[accept] (AS25) [below =8.5cm of AS24]
  {\Huge{$\varocircle$}};

\node[deadend] (AS26) [below right =4cm of AS11]
  {\Huge{$\varotimes$}};

  \path[very thick] (ASTART) edge node  {} (AS1);
  \path[very thick] (AS1) edge node  {} (AS2);
  \path[very thick] (AS2) edge node  {} (AS3);
  \path[very thick] (AS3) edge node [above] {\hspace{-0.7cm}$0.01$} (AS4);
  \path[very thick] (AS3) edge node [above] {\hspace{0.9cm}$0.99$} (AS9);
  \path[very thick] (AS4) edge node  {} (AS5);
  \path[very thick] (AS4) edge node  {} (AS6);
  \path[very thick] (AS5) edge node  {} (AS7);
  \path[very thick] (AS6) edge node  {} (AS8);
  \path[very thick] (AS7) edge[bend left=80]  node [left] {$1$} (AS1);
  \path[very thick] (AS8) edge node  [right] {$1$} (AS24);
  \path[very thick] (AS9) edge node  {} (AS10);
  \path[very thick] (AS9) edge node  {} (AS11);
  \path[very thick] (AS10) edge [bend right=10] node   {} (AS26);
  \path[very thick] (AS10) edge node  {} (AS13);
  \path[very thick] (AS11) edge [bend left=10] node  {} (AS26);
  \path[very thick] (AS11) edge node  {} (AS12);
  \path[very thick] (AS12) edge node [right]  {$1$} (AS23);
  \path[very thick] (AS13) edge node [right] {$0.01$} (AS15);
  \path[very thick] (AS15) edge node  {} (AS16);
  \path[very thick] (AS15) edge node  {} (AS17);
  \path[very thick] (AS16) edge node  {} (AS18);
  \path[very thick] (AS17) edge node  {} (AS20);
  \path[very thick] (AS20) edge node [left] {$1$} (AS21);
  \path[very thick] (AS21) edge node  {} (AS22);
  \path[very thick] (AS21) edge node  {} (AS21a);
  \path[very thick] (AS21a) edge node  {} (AS26);
  \path[very thick] (AS21a) edge node  {} (AS21b);
  \path[very thick] (AS21b) edge node [right] {\hspace{0.4cm}$1$} (AS21);
  \path[very thick] (AS22) edge [loop below] node  {} (AS22);
  \path[very thick] (AS22) edge [bend left=10] node  {} (AS25);
  \path[very thick] (AS13) edge [bend right=80] node [above] {\hspace{1.2cm}$0.99$} (AS9);
  \path[very thick] (AS18) edge[bend right=85] node [above] {\hspace{0.3cm}$1$} (AS9);
  \path[very thick] (AS24) edge node  {} (AS24a);
  \path[very thick] (AS24a) edge node  {} (AS25);
  \path[very thick] (AS23) edge node  {} (AS23a);
  \path[very thick] (AS23a) edge node  {} (AS23b);
  \path[very thick] (AS23b) edge [bend right=60] node  {\hspace{0.8cm}$1$} (AS23);
  \path[very thick] (AS16) edge [bend left=95] node  {} (AS26);
  \path[very thick] (AS21a) edge [bend left=10] node  {} (AS25);

\end{tikzpicture}
}
 \caption{Abstraction of the program from Fig.~\ref{fig:ex1_code}.}
\label{fig:ex1_abstr}
\end{figure}


The probability of reaching \texttt{fail} depends on the behaviour of the environment. 
The smallest (largest) probability clearly corresponds 
to the environment always choosing \verb+goto 1+ (\verb+fail+), and its value is $0$ ($0.01$). 
However, a brute force automatic technique will construct the standard
semantics of the program, a Markov decision process (MDP) with over 400 states, part of which
is shown in Fig.~\ref{fig:ex1_code}. 
\andreas{Hinweis auf tight intervals eingefuegt}
We informally introduce 
our abstraction technique, which does better and is able to infer tight intervals for both the smallest and the largest probability. It is based on the parallel 
or menu-based predicate abstraction of
\cite{wachter,bjoern-phd}, which we adapt to arbitrary abstract domains.

\subsection{Constructing a valid abstraction} \label{sec:ex_arena_constr}
Given an abstract domain and a widening operator, we abstract
the MDP of the program into four different {\em stochastic 2-Player games} sharing the same 
arena and the same rules (i.e., the games differ only on the winning conditions). A round of the game starts at an abstract
state $n$ (we think of $n$ as a set of concrete states)
with Player 1 to move. Let $n_i$ be the set of concrete states of $n$ that enable command $A_i$. Player 1 proposes an $A_i$ such that $n_i \neq \emptyset$,
modeled by a move from $n$ to a node $\tup{n , A_i}$. If $n$ contains some 
final state, then Player 1 can also propose to end the play (modeled by a move to $\tup{n, \varocircle}$). Then it is Player 2's turn. If Player 1 proposes $A_i$, then Player 2 can 
accept the proposal (modeled by a move to a node determined below), or
reject it and end the play (modeled by a move to another distinguished node $\varotimes$), but only 
if $n_i \not= n$. If Player 2 accepts $A_i$, she moves to some node $\tup{n, A_i, n'}$ such that 
$n' \subseteq n_i$, i.e., Player 2 can ''pick`` a subset $n'$ of $n_i$ out of the subsets offered by the game arena (every
concrete state in $n_i$ is contained in one such $n'$).

The next node of the play is determined 
probabilistically: one of the updates of $A_i$ is selected randomly according to the probabilities,
and the play moves to the abstract state obtained by applying the 
update and (in certain situations) the widening operator to $n'$. 
If Player 1 proposes $\varocircle$ by choosing $\tup{n, \varocircle}$, then Player 2 can accept the proposal, 
(modeled by a move to $\varocircle$) or, if not all concrete states of $n$ are final, 
reject it (modeled by a move $\tup{n,\varocircle} \rightarrow \tup{n,\varocircle}$).

\andreas{Beschreibung eingefuegt, bitte anschauen}

Fig.~\ref{fig:ex1_abstr} shows an arena for 
the program of Fig.~\ref{fig:ex1_code} with the abstract domain and widening operator described in the following. Nodes owned by Player 1 are drawn 
as white rectangles, nodes owned by Player 2 as circles, and probabilistic nodes
as grey rectangles. In the figure we label a node $\tup{n , A_i}$ belonging to 
Player 2 with $A_i$ and a probabilistic node $\tup{n , A_i, n'}$ with $n'$ 
($n$ resp. $n$ and $A_i$ can easily be reconstructed by inspecting the direct 
predecessors). Nodes of the form $\tup{n, \varocircle}$ are labeled 
with '$\varocircle?$'.

A (concrete) state of the example program is a pair
$\langle {\it ctr}, {\it nrp} \rangle$, and an abstract state is a pair $\langle {\it ctr}, [a, b] \rangle$, 
where $[a,b]$ is an interval of values of \texttt{nrp} (i.e., \texttt{ctr} is not abstracted in the example). 
The widening operator $\nabla$ works as follows:
if the abstract state $\tup{{\it ctr}, [a,b]}$ has an ancestor
$\tup{{\it ctr}, [a',b']}$ along the path between it and the initial state given by the construction, then we overapproximate $\tup{{\it ctr}, [a,b]}$ by $\tup{{\it ctr}, s}$, with 
$s =  [a',b']\, \nabla \, [\min(a, a'), \max(b,b')]$. For instance
the node $n = \tup{1, [0,\infty)}$ in Fig.~\ref{fig:ex1_abstr}
enables $\texttt{A1}$ and $\texttt{A2}$, and so it has two successors. Since $n$ contains concrete states
that do not enable $\texttt{A1}$ and concrete states that do not enable $\texttt{A2}$, both of them have
$\varotimes$ as successor. The node $\tup{n, \texttt{A1}, \tup{1, [0,99]}}$ \andreas{EINGEFUEGT:}
(whose label is abbreviated by $\tup{1, [0,99]}$ in the figure) is probabilistic. Without widening, its successors would be 
$\tup{2, [0,99]}$ and $\tup{1, [1,100]}$ with probabilities $0.01$ and $0.99$. However,
$\tup{1, [1,100]}$ has $\tup{1, [0, \infty)}$ as (direct) predecessor, which has the same $ctr$-value. Therefore
the widening overapproximates $\tup{1, [1,100]}$ to $\tup{1, [0, \infty) \nabla [0,\infty)} = \tup{1, [0, \infty)}$,
and hence we insert an edge from $\tup{n, \texttt{A1}, \tup{1, [0,99]}}$ (back) to $\tup{1, [0,\infty)}$, labeled by
$0.99$.


After building the arena, we compute lower and upper bounds for the 
minimal and maximal reachability probabilities as the {\em values}
of four different games, defined as the winning probability of Player 1 
for optimal play. The winning conditions of the games are as follows: 
\begin{itemize}
 \item[(a)] Lower bound for the maximal probability: Player 1 wins
       if the play ends by reaching $\varocircle$, otherwise Player 2 wins. 
 \item[(b)] Upper bound for the maximal probability: Players 1 and 2 both win
       if the play ends by reaching $\varocircle$, and both lose otherwise. 
 \item[(c)] Lower bound for the minimal probability: Players 1 and 2 both lose
       if the play ends by reaching $\varocircle$ or $\varotimes$, and both win otherwise. 
 \item[(d)] Upper bound for the minimal probability: Player 1 loses if
       the play ends by reaching $\varocircle$ or $\varotimes$, otherwise Player 2 loses.
\end{itemize}
For the intuition behind these winning conditions, consider first game (a).
Since Player 1 models the environment and wins by reaching $\varocircle$,
the environment's goal is to reach a final state. Imagine first that the abstraction is the trivial one, i.e.,
abstract and concrete states coincide. Then Player 2 never has a choice, and the optimal strategy 
for Player 1 determines a set $S$ of action sequences whose total probability is equal to the 
maximal probability of reaching a final state. Imagine now that the abstraction is coarser. In the arena for 
the abstract game the sequences of $S$ are still possible, but now Player 2
may be able to prevent them, for instance by moving to $\varotimes$ 
when an abstract state contains concrete states not enabling the next action in the sequence. Therefore, in the 
abstract game the probability that Player 1 wins 
can be at most equal to the maximal probability. In game (b) the team formed by the two players
can exploit the spurious paths introduced by the abstraction to find a strategy leading to a better 
set of paths; in any case, the probability of $S$ is a lower bound for the winning probability
of the team. The intuition behind games (c) and (d) is similar.

In our example, optimal strategies for game (b) are:
for Player 1, always play the ``rightmost'' choice, except at $\tup{2, [0,99]}$, where she should play $\texttt{A4}$;
for Player 2, play $\varocircle$ if possible, otherwise anything but $\varotimes$.
The value of the game is $1$. In game (a), the optimal strategy
for Player 1 is the same, whereas Player 2 always plays $\varotimes$ (resp. stays in $\varocircle?$) whenever
possible. The value of the game is $0.01$. We get $[0.01, 1]$ as lower and upper
bound for the maximal probability.  
For the minimal probability we get the trivial bounds $[0, 1]$.


To get more precision, we can 
skip widenings at certain situations during the construction. If we e.g. apply widening only after the second unrolling of the loop,
the resulting abstraction allows us to obtain  the more precise bounds $[0, 0.01]$ and $[0.01, 0.01]$ for minimal and maximal reachability, respectively.

\begin{figure}[tbp]
\centering
\begin{minipage}[t]{0.45\linewidth} 
 \begin{verbatim}
     int c = 0, i = 0;
  1: if choice(0.5) then
  2:   while (i <= 100) 
  3:     i = i+1;
  4:     c = c-i+2 
  5: if (c >= i) then fail
 \end{verbatim}
\end{minipage}
\begin{minipage}[t]{0.45\linewidth} 
 \begin{verbatim}     
     int c = 0, i = 0;
  1: while(i <= 100) 
  2:   if choice(0.5) then i = (i+1);
  3:   c = c-i+2;
  4: if (c >= i) then fail
\end{verbatim}
\end{minipage}
\caption{Example programs 2 and 3.}
\label{fig:prog2}
\end{figure}

The main theoretical result of our paper
is the counterpart of the results of \cite{kattenbelt,wachter}: for arbitrary
abstraction domains, the values of the four games described above indeed 
yield upper and lower bounds of the maximal and minimal probability of 
reaching the goal nodes. 

In order to give a first impression of the advantages of 
abstraction domains beyond predicate abstraction in the probabilistic case, 
consider the (deterministic) pseudo code on the left of Fig.~\ref{fig:prog2}, 
a variant of the program above. Here $\texttt{choice}(p)$ models a 
call to a random number generator that returns $1$ with 
probability $p$ and $0$ with probability $1-p$. 

It is easy to see that $c \leq 1$ is a global invariant, and so the probability of failure is exactly $0.5$. 
Hence a simple invariant like $c \leq k$ for a $k \leq 100$, together with the postcondition $i > 100$ of the loop 
would be sufficient to negate the guard of the statement at line 5.
However, when this program is analysed with PASS \cite{wachter,hermanns}, a leading tool on probabilistic 
abstraction refinement on the basis of predicate abstraction, the while loop is 
unrolled 100 times because the tool fails to ``catch'' the invariant, independently of the options chosen to refine the abstraction
\footnote{Actually, the input language of PASS does not explicitly include
while loops, they have to be simulated. But this does not affect the
analysis.}. 

On the other hand, an analysis of
the program with the standard interval domain, the standard widening operator, and the standard technique of delaying
widenings \cite{astree}, easily `catches'' the invariant (see Section \ref{sec:experiments}).
The same happens for the program on the right of the figure, which exhibits a more interesting probabilistic behaviour, especially
a probabilistic choice within a loop: we obtain good upper and lower bounds for the probability of failure using the standard interval domain.
Notice that examples exhibiting the opposite behaviour 
(predicate abstraction succeeds where interval analysis fails) are not difficult to find; our thesis is only that the game-based abstraction 
approach of \cite{wachter,kattenbelt} can be extended to arbitrary 
abstract domains, making it more flexible and efficient.

%% file: abstraction_game-new.tex
\section{Stochastic 2-Player Games}
\label{sec:preliminaries}

This section introduces stochastic 2-Player games. For a more thorough introduction into the
subject and proofs for the theorems see e.g.~\cite{puterman,condon1,condon2}. 

Let $S$ be a countable set.
We denote by $\dist(S)$ the set of all distributions $\delta: S \rightarrow [0,1]$ over $S$ with $\delta(x) = 0$ for all but finitely many $x \in S$.

\begin{definition}
A \emph{stochastic 2-Player game} $\sG$ (short 2-Player game) is a tuple 
$((V_1, V_2, V_p), E, \delta, s_0)$, where
 \begin{itemize}
 \item $V_1, V_2, V_p$ are distinct, countable sets of states. We set $V = V_1 \cup V_2 \cup V_p$;
 \item $E \subseteq (V_1 \cup V_2) \times V$ is the set of \emph{admissible} player choices;
 \item $\delta: V_p \rightarrow \dist(V)$ is a probabilistic transition function;
 \item $s_0 \in V_1$ is the start state.
\end{itemize}
  Instead of $(q,r) \in E$ we often write $q \rightarrow r$.
  A string $w \in V^+$ is a \emph{finite run} (short: run) of $\sG$  if 
  (a) $w = s_0$, or (b) $w = w' s' s$ for some run $w's' \in V^\ast(V_1 \cup V_2)$ and $s' \rightarrow s$, or
  (c) $w = w' s' s$ for some run $w's' \in V^\ast V_p$ such that $\delta(s')(s) > 0$. We denote the set of all runs of $\sG$ by $\cyl{\sG}$. 
  A run $w = x_1 \ldots x_k$ is {\em accepting relative to $F \subseteq V_1$} if $x_k \in F$ and $x_i \not\in F$ for $1 \leq i < k$.
  The set of accepting runs relative to $F$ is denoted by $\cyl{\sG, F}$.

A stochastic 2-Player game with $V_2 = \emptyset$ is called a  \emph{Markov Decision Process (MDP)}, and then we write $\sG = ((V_1, V_p), E, \delta, s_0)$.
\end{definition}

Fix for the rest of the section a 2-Player game $\sG = ((V_1, V_2, V_p), E, \delta, s_0)$.
The behaviours of Player 1 and 2 in $\sG$ are described with the help of \emph{strategies}:
\begin{definition}\mbox{}
  A strategy for Player $i\in \{1, 2\}$ in $\sG$ is a partial function 
   $\phi: \cyl{\sG} \rightarrow \dist(V)$ satisfying the following two conditions:
  \begin{itemize}
   \item $\phi(w)$ is defined if{}f $w=w' v \in V^\ast V_i$ and
   $v \rightarrow x$ for some $x \in V$; and
   \item if $\phi(w)$ is defined and $\phi(w)(x) > 0$ then $wx$ is a run.  
  \end{itemize}
\noindent We denote the set of strategies for Player $i$ by $S_i(\sG)$. 
 A strategy $\phi$ is \emph{memoryless} if $\phi(w_1) = \phi(w_2)$
 for any two runs $w_1, w_2$ ending in the same node of $V_i$, and \emph{non-randomized} if for every run $w$ such that $\phi(w)$ is defined there is a node $x$ such that $\phi(w)(x) = 1$. Given strategies $\phi_1, \phi_2$ for 
 Players 1 and 2, the {\em value} $\val{w}{}_{\sG[\phi_1, \phi_2]}$ of a run $w$ under $\phi_1, \phi_2$ is defined as follows:
  \begin{itemize}
   \item If $w = s_0$, then $\val{w}{}_{\gpp} = 1$.
   \item If $w = w' s \in V^\ast V_i $ for $i \in \{1,2\}$ and  $\phi_i(w')$ is defined,
    then $\val{w}{}_{\gpp} = \val{w'}{}_{\gpp}\cdot \phi_i(w')(s)$.
   \item If $w = w' s' s$ for some run $w's' \in V^\ast V_p$ then \\$\val{w}{}_{\gpp} = \val{w's'}{}_{\gpp} \cdot \delta(s')(s)$.
   \item Otherwise $\val{w}{}_{\gpp} = 0$.
 \end{itemize}
\end{definition}
\andreas{Definition in 2+3 aufgeteilt}
We are interested in \emph{probabilistic reachability}:
\begin{definition}
The \emph{probability} $\reach{\gpp, F}$ \emph{of reaching $F \subset V_1$} in $\sG$ under
strategies $\phi_1$ and $\phi_2$ of Players 1 and 2 is
\begin{equation*}
 \reach{\gpp, F}{} := \sum_{w \in \cyl{\sG, F}}{\val{w}{}_{\sG[\phi_1, \phi_2]}}.
\end{equation*}
\end{definition}
If the context is clear, we often omit the subscript of $\val{\cdot}{}$.
We write $\cyl{\gpp}$ (resp. $\cyl{\gpp, F}$) for the set of all finite runs $r \in \cyl{\sG}$ (resp. $r \in \cyl{\sG, F}$)
with $\val{r}{}_{\gpp} > 0$.
In a MDP $\sM$ we do not require to have a strategy for the second Player. Here we just write 
$\reach{\sM[\phi_1], F}$ for a given strategy $\phi_1 \in S_1(\sM)$. 
\begin{definition}
 Let $\sG = ((V_1, V_2, V_p), E, \delta, s_0)$ be a 2-Player game, and $F \subset V_1$. The \emph{extremal game values} $\reach{\sG, F}^{++}, \reach{\sG, F}^{+-}, \reach{\sG, F}^{-+}$ and $\reach{\sG, F}^{--}$ are
\begin{align*}
  \reach{\sG, F}^{++} &:= \sup_{\phi_1 \in S_1(\sG)} \sup_{\phi_2\in S_2(\sG)} \reach{\gpp, F}{} \\
  \reach{\sG, F}^{+-} &:= \sup_{\phi_1 \in S_1(\sG)} \inf_{\phi_2\in S_2(\sG)} \reach{\gpp, F}{} \\
  \reach{\sG, F}^{-+} &:= \inf_{\phi_1 \in S_1(\sG)} \sup_{\phi_2\in S_2(\sG)} \reach{\gpp, F}{} \\
  \reach{\sG, F}^{--} &:= \inf_{\phi_1 \in S_1(\sG)} \inf_{\phi_2\in S_2(\sG)} \reach{\gpp, F}{}
\end{align*} 
If $\sG$ is a MDP, we define $\reach{\sG, F}^{+} := \reach{\sG, F}^{++} = \reach{\sG, F}^{+-}$ and 
$\reach{\sG, F}^{-} := \reach{\sG, F}^{--} = \reach{\sG, F}^{-+}$.
\end{definition}

The following well-known theorem will be crucial for the validity of our abstractions~\cite{condon1}:
\begin{theorem}
Let $F \subset V_1$.
 For each $\kappa \in \{++,+-,-+,--\}$ there exist \emph{non-randomized and memoryless} strategies
 $\phi_1^\kappa \in S_1(\sG), \phi_2^\kappa \in S_2(\sG)$ such that 
\begin{equation*}
\reach{\sG, F}^{\kappa} = \reach{\sG[\phi_1^\kappa, \phi_2^\kappa], F}. 
\end{equation*}
\end{theorem}
Extremal game values can be computed e.g. by variants of value iteration~\cite{condon2}.

\section{Abstractions of Probabilistic Programs}
\label{sec:abstraction}
We start by giving a formal definition of NPPs.

  \begin{definition}
  Let $\sV$ be a finite set of variables, where $x \in \sV$ has a range $\text{rng}(x)$.
  A \emph{configuration} (or state) of $\sV$ is a map $\sigma \colon \sV \rightarrow \bigcup_{x \in \sV}{\rng(x)}$ such that $\sigma(x) \in \rng(x)$ for all $x \in \sV$.  The set of all configurations is denoted by $\Sigma_\sV$. A {\em transition} is a map $f \in 2^{\Sigma_\sV} \rightarrow 2^{\Sigma_\sV}$  such that  $\abs{f(\{\sigma\})} \leq 1$ 
  for all $\sigma \in \Sigma_\sV$ (i.e., a transition maps a single configuration
  to the empty set or to a singleton again), and
  \begin{equation*} 
   \bigcup_{\sigma \in M}{f(\{ \sigma \})} = f(M) \text{ for all } M \subseteq \configs.
  \end{equation*}

  A transition $g$ is a \emph{guard} if $g(\{\sigma\}) \in \{ \{\sigma\},  \emptyset\}$ for every configuration $\sigma$. We say that $\sigma$ {\em enables} $g$ if $g(\{\sigma\}) = \{\sigma\}$.
  A transition $c $ is an \emph{assignment} if $\abs{c(\{\sigma\})} = 1$  for all $\sigma \in \Sigma_\sV$. The semantics of an assignment $c$ is the map $\sem{c}: \configs \rightarrow \configs$ given by $\sem{c}(\sigma) := \sigma'$ if $c(\{\sigma\}) = \{\sigma'\}$.
  The set of transitions is denoted by $\trans$.
 \end{definition}

\begin{definition}{Nondeterministic Probabilistic Programs.}\\
A \emph{nondeterministic probabilistic program} (NPP) is a triple $P = (\sV, \sigma_0, \sC)$ where 
  $\sV$ is a finite set of program {\em variables}, $\sigma_0 \in \Sigma_\sV$ is the {\em initial configuration}, 
  and $\sC $ is a finite set of \emph{guarded commands}.
  A \emph{guarded command} $A$ has the form 
  $A = g \rightarrow p_1: c_1 + \ldots + p_m: c_m$, where $m \geq 1$, 
  $g$ is a guard, $p_1, \ldots, p_m$ are probabilities adding up to $1$,
  and $c_1, \ldots, c_m$ are assignments. We denote the guard of 
  $A$ by $g_A$, the {\em updates} $\{\tup{p_1, c_1}, \ldots \tup{p_m, c_m}\}$
  of $A$ by $up_{A}$, and the set $\{up_{A} \mid A \in \sC\}$ by $up_\sC$.
 \end{definition}

\begin{definition}{Semantics of NPPs and Reachability Problem.}\\
The {\em MDP associated to} a NPP $P = (\sV, \sigma_0, \sC)$ is 
$\sM_P = ((V_1, V_p), E, \delta, \sigma_0)$, where
$V_1 = \Sigma_\sV$, $V_p =  \Sigma_\sV \times \sC$, $E \,\subseteq V_1 \times (V_1 \cup V_p)$, $\delta: (\Sigma_\sV \times \sC)  \rightarrow \dist(V_1)$, and  
for every $A \in \sC$, $\sigma, \sigma' \in \Sigma_\sV$
\begin{equation*}
  \sigma \rightarrow \tup{\sigma, A} \text{ iff } \sigma \text{ enables } g_A \text{ and } 
  \delta(\tup{\sigma, A})(\sigma') := \sum_{\tup{p,c} \in up_A:\ 
  \sem{c}(\sigma) = \sigma'}{p}\ .
\end{equation*}
\noindent The {\em reachability problem} for $P$ relative to a set $F \subseteq \Sigma_\sV$ of states such that
$\sigma_0 \notin F$ is the problem of computing $\reach{\sM_P, F}^{+}$ and $\reach{\sM_P, F}^{-}$. We call $F$ the set of {\em final states}.
\end{definition}
\noindent We assume in the following that for every run $w \sigma \in V^\ast V_1$ in $\sM_P$ 
either $\sigma \in F$ or $\sigma$ enables the guard of at least one command (i.e., we do not 'get stuck' during the computation). 
This can e.g. be achieved by adding a suitable guarded command that simulates a self loop.

\subsection{Abstracting NPPs}
\label{subsec:abstractnpp}
We abstract NPPs using the Abstract Interpretation framework (see~\cite{cousot}). 
As usual, an abstract domain is a complete lattice
$(D^\sharp, \sqsubseteq, \top, \bot, \sqcup, \sqcap)$ (short $D^\sharp$), and we assume the existence of
monotone abstraction and concretization 
maps $\alpha: 2^{\Sigma_\sV} \rightarrow D^\sharp$ and
$\gamma: D^\sharp \rightarrow 2^{\Sigma_\sV}$ forming a Galois connection between $D^\sharp$ and $2^{\Sigma_V}$.
A \emph{widen operator} is a mapping $\nabla: D^\sharp \times D^\sharp \rightarrow D^\sharp$ 
satisfying (i) $a \nabla b \sqsupseteq a$ and $a \nabla b \sqsupseteq b$ for all $a,b \in D^\sharp$,
and (ii) for every strictly increasing sequence  $a_0 \sqsubset a_1 \sqsubset \ldots $ in $D^\sharp$
the sequence $(b_i)^{i \in \mathbb{N}}$ defined by $b_0 = a_0$ and $b_{i+1} = b_i \, \nabla \, a_{i+1}$ is
stationary.

We abstract sets of configurations by elements of $D^\sharp$.
Following ideas from~\cite{kattenbelt,wachter,hermanns,bjoern-phd}, the abstraction of an 
NPP is a 2-player stochastic game. We formalize which games are \emph{valid abstractions} of a given NPP
(compare the definition to the comments in Section~\ref{sec:ex_arena_constr}):

\begin{definition}
\label{def:validabstraction}
Let  $P = (\sV, \sigma_0, \sC)$ be a NPP with a set $F \subseteq \Sigma_\sV$ of
final states such that $\sigma_0 \notin F$.
A 2-player game $\sG = ((V_1, V_2, V_p), E, \delta, s_0)$ with finitely many nodes is a \emph{valid abstraction} of $P$ relative to $F$ for $D^\sharp$ if 
\begin{itemize}
 \item $V_1$ contains a subset of $D^\sharp$ plus two distinguished states
       $\varocircle, \varotimes$;
 \item $V_2$ is a set of pairs $\langle s, A \rangle$, where $s \in V_1 \setminus \{\varocircle, \varotimes\}$ and either 
       $A = \varocircle$ or $A$ is a command of $\sC$ enabled by some state of $\gamma(s)$;
 \item $V_p$ is a set of fourtuples $\langle s, A, s', d \rangle$, where $s, s' \in V_1 \setminus \{\varocircle, \varotimes\}$ 
       such that $s \sqsupseteq s'$, $A$ is a command enabled by some state of $\gamma(s')$, and 
       $d$ is the mapping that assigns to every update $\langle p, c \rangle \in up_A$ an abstract state $s' \in V_1$ with $\gamma(d(\langle p, c \rangle)) \supseteq c(\gamma(s'))$;
 \item $s_0 = \alpha(\{\sigma_0\})$;
\end{itemize}
\noindent and the following conditions hold:
\begin{enumerate}
\item For every $s \in V_1 \setminus \{\varocircle, \varotimes\}$ and every $A \in \sC$:
\begin{itemize}
 \item[(a)] If $\gamma(s) \cap F \not=\emptyset$ then
       $s \rightarrow \tup{s, \varocircle} \rightarrow \varocircle$. If moreover
       $\gamma(s) \subseteq F$, then $\tup{s, \varocircle}$ is the only successor of $s$; otherwise, also
       $\tup{s, \varocircle} \rightarrow \tup{s, \varocircle}$ holds.
       \\
       (*
	If $\gamma(s)$ contains some final state, then Player 1 can propose $\varocircle$.
        If all states of $\gamma(s)$ are final, then Player 2 must accept, otherwise it can
        accept, or reject by staying in $\tup{s, \varocircle}$.
       *)
 \item[(b)] If  $g_A(\gamma(s)) \not=\emptyset$ then $\tup{s, A} \in V_2$ and $s \rightarrow \tup{s, A}$. 
       \\
       (*
       If some state of $\gamma(s)$ enables $A$ then Player 1 can propose $A$.
       *)
\end{itemize}
\item For every pair $\tup{s, A} \in V_2$ and every $A \in \sC$:
\begin{itemize}
 \item[(a)] there exist nodes $\{\tup{s, A, s_1, d_1}, \ldots, \tup{s, A, s_k, d_k}\} \subseteq V_p$
   such that $\tup{s,A} \rightarrow \tup{s, A, s_i, d_i}$ for every $i \leq k$ and $g_A(\gamma(s)) \subseteq \bigcup_{j=1}^k \gamma(s_j)$.
   \\
   (*
     If Player 2 accepts $A$, then she can
     pick any concrete state $\sigma \in \gamma(s)$ enabling $A$, and choose a successor
     $\tup{s, A, s_i, d_i}$ such that $\sigma \in \gamma(s_i)$.
   *)
 \item[(b)] If $\gamma(s) \cap F \not= \emptyset$, then $\tup{s, A} \rightarrow \varocircle$.\\
 (* If $\gamma(s)$ contains some final state, then Player 2 can reject and move to $\varocircle$. *)
 \item[(c)] If $g_A(\gamma(s)) \not= \gamma(s)$, then $\tup{s, A} \rightarrow \varotimes$. \\
 (* If some state of $\gamma(s)$ does not enable $A$, then Player 2 can reject $A$ and 
    move to $\varotimes$. *)
\end{itemize}
\item For every $\tup{s,A,s',d} \in V_p$ and every abstract state $s'' \in V_1$:
\begin{equation*}
  \delta(\tup{s, A, s', d})(s'') := \sum_{\tup{p,c} \in up_A \colon d(\tup{p,c}) = s''}{p} \ .
\end{equation*}
\item The states $\varotimes$ and $\varocircle$ have no outgoing edges.
\end{enumerate}
\end{definition}

We can now state the main theorem of the paper:
the extremal game values of the games derived from valid abstractions provide upper and lower bounds
on the maximal and minimal reachability probabilities. The complete proof is given in~\iftechrep{ the appendix}{\cite{EG:sasTechRep}}.

\begin{theorem}
\label{thm:main}
 Let $P$ be a NPP and let $\sG$ be a valid abstraction of $P$ relative to $F$ for the abstract domain
$D^\sharp$. Then 
\begin{align*}
 \reach{\sM_P, F}^{-} &\in [\reach{\sG, \{\varocircle, \varotimes\}}^{--}, \reach{\sG, \{\varocircle, \varotimes\}}^{-+}] \text{ and } \\
 \reach{\sM_P, F}^{+} &\in [\reach{\sG, \{\varocircle\}}^{+-}, \reach{\sG, \{\varocircle\}}^{++}].
\end{align*}
\end{theorem}

\begin{proof} 
({\it Sketch.)} The result is an easy consequence of the following three assertions:
\begin{itemize}
\item[(1)] Given a strategy $\phi$ of the (single) player in $\sM_P$, there exists a strategy $\phi_1 \in S_1(\sG)$ such that
 \begin{equation*}
   \inf_{\psi \in S_2(\sG)} \reach{\sG[\phi_1, \psi], \{\varocircle, \varotimes\}} \leq \reach{\sM_P[\phi], F} \text{ and } 
  \end{equation*}
\begin{equation*}
 \sup_{\psi \in S_2(\sG)} \reach{\sG[\phi_1, \psi], \{\varocircle\}} \geq \reach{\sM_P[\phi], F}  .
\end{equation*}
\item[(2)] Given a strategy $\phi_1 \in S_1(\sG)$ there exists a strategy $\phi \in S_1(\sM_P)$ such that
\begin{equation*}
     \reach{\sM_P[\phi], F}{} \leq \sup_{\psi \in S_2(\sG)}  \reach{\sG[\phi_1, \psi], \{\varotimes, \varocircle\}}{}.
\end{equation*}
\item[(3)] Given a strategy $\phi_1 \in S_1(\sG)$ there exists a strategy $\phi \in S_1(\sM_P)$ such that
\begin{equation*}
   \reach{\sM_P[\phi], F}{} \geq \inf_{\psi \in S_2(\sG)}  \reach{\sG[\phi_1, \psi], \{\varocircle\}}{}.
\end{equation*}
\end{itemize}
\noindent To prove (1) (the other two assertions are similar), 
we use $\phi$ to define a function $D$
that distributes the probabilistic mass of a run $R \in \cyl{\sG}$
among all the runs $r \in \cyl{\sM}$ (where $\sM$ is a normalization of $\sM_P$). 
The strategies
$\phi_1$ and $\phi_2$ are then chosen so that they produce the same distribution, i.e., the mass of all the runs $r$ that follow the strategies and correspond to an abstract run $R$ following $\phi$
is equal to the mass of $R$.
\qed
\end{proof}
Recall that in predicate abstraction the 
concretizations of two abstract states are disjoint sets of configurations
(disjointness property). 
This allows to easily define a Galois connection between the sets of
functions assigning values to the abstract and the concrete states:
Given a concrete valuator $f$, its abstraction is the function that assigns
to a set $X$ the minimal resp. the maximal value assigned by $f$ to
the elements of $X$. Here we 
have to distribute the value of a concrete state 
into multiple abstract states (which is what we do in our proof).

\subsection{An Algorithm for Constructing Valid Abstractions}
Algorithm~\ref{alg:compute-g} builds a valid abstraction $\sG$ of a NPP $P$ relative 
to a set $F$ of final states for a given abstract domain $D^\sharp$.
It is inspired by the algorithms of \cite{blast,gulavani} for constructing abstract 
reachability trees. It constructs the initial state 
$s_0 = \alpha(\{\sigma_0\})$ and generates transitions and successor states 
in a breadth-first fashion using a work list called $\worklist$. The 
\texttt{GENSUCCS} procedure constructs the successors of a node guided by the rules from Def.~\ref{def:validabstraction}.
It uses abstract transformers $g^\sharp$ and $c^\sharp$ for the guards and commands 
of the NPP. Hereby a transformer $g^\sharp: D^\sharp \rightarrow 2^{D^\sharp}$ abstracting a guard $g$ has to satisfy that
{for all } $a \in D^\sharp$, $g^\sharp(a)$ is finite and $\bigcup_{b \in g^\sharp(a)}{\gamma(b)} \supseteq g(\gamma(a))$.
Allowing $g^\sharp$ to return a set rather than just one element from $D^\sharp$ can help increasing the accuracy of $\sG$. Here we implicitly make use of abstract powerset domains.
\texttt{GENSUCCS} assumes that it can be decided whether $\gamma(s) \cap F = \emptyset$,
$\gamma(s) \not\subseteq F$, $g_A(\gamma(s)) \neq \emptyset$ or $g_A(\gamma(s)) \neq \gamma(s)$ hold (lines 2 and 4). The assumptions on $F$ are reasonable, since in most cases the set $F$ has a very simple
shape, and could be replaced by conservative tests on the abstract. 
A conservative decision procedure suffices for the test $g_A(\gamma(s)) \not= \gamma(s)$, 
\andreas{(similarly with  $g_A(\gamma(s)) \neq \emptyset$) eingefuegt} with the only requirement that
if it returns 0, then  $g_A(\gamma(s)) = \gamma(s)$ has to hold. The same holds for the test $g_A(\gamma(s)) \neq \emptyset$.
\texttt{GENSUCCS} closely follows the definition of a valid abstraction, as specified in Def.~\ref{def:validabstraction}. 

Lines 1 and 2 guarantee that condition (1a) of Def.~\ref{def:validabstraction} holds, and, 
similarly, line 4 guarantees condition (1b). Similarly, 
lines 3 and 5 are needed to satisfy conditions (2b) and (2c), respectively.
The loop at line 6 generates the nodes of the form $\tup{s,A,s_i,d}$ required by condition (2a) of our definition, 
and the loop at line 7 constructs the function $d$ appearing in  condition 3.

As usual, termination of the algorithm requires to use widenings. This is the role
of the \texttt{EXTRAPOLATE} procedure. During the construction, we use the function $\pred(\cdot)$ 
to store for every node $s \in V_1\setminus\{s_0,\varocircle, \varotimes\}$ 
its predecessor in the spanning tree induced by the construction (we call it \emph{the spanning tree} from now on). 
For a node $s' \in V_1$ that was created as the result of 
chosing a guarded command $A$, the procedure finds the nearest predecessor $s$
in the spanning tree with the same property, and uses $s$ to perform a widen operation. 
\andreas{Etwas umgeschrieben}
Note that in the introductory example, another strategy was used: There we applied widenings only for states with matching control location. 
The strategy used in \texttt{EXTRAPOLATE} does not use additional information like control flow and thus can be used for arbitrary NPPs.

\noindent We can now prove (see\iftechrep{~the appendix}{~\cite{EG:sasTechRep}}):
\begin{theorem}
\label{thm:termination}
Algorithm~\ref{alg:compute-g} terminates, and its result $\sG$ is a valid abstraction.
\end{theorem}

\begin{algorithm}[htb]
\label{alg:compute-g}
\DontPrintSemicolon
\KwIn{NPP $P = (\sV, \sigma_0, \sC)$, 
abstract domain $D^\sharp$, set of final states $F \subseteq \Sigma_{\sV}$, 
widening $\nabla$.}
\KwOut{2-Player game $\sG = ((V_1, V_2, V_p), E, \delta, s_0)$.}

\vspace{0.2cm}
$s_0 = \alpha(\{\sigma_0\}); \; \text{pred}(s_0) \leftarrow \text{nil}$\;
$V_1 \leftarrow \{s_0, \varotimes, \varocircle\}; \; V_2 \leftarrow \emptyset; \; V_p \leftarrow \emptyset;\;\worklist \leftarrow \{s_0\}$\;
\While{$\worklist \not= \emptyset$} {
  Remove $s$ from the head of $\worklist$; $\texttt{GENSUCCS}(s)$
}\;
{\bf Procedure} \texttt{GENSUCCS}($s \in V_1)$\;

 $\fopt \leftarrow {\it false}$\;
 \If{$\gamma(s) \cap F \not= \emptyset$} {
 \nl $E \leftarrow E \cup \{(s, \tup{s, \varocircle}), (\tup{s, \varocircle}, \varocircle)\}$\;
 \nl \lIf{$\gamma(s) \not\subseteq F$} {
    \{ $E \leftarrow E \cup \{(\tup{s, \varocircle}, \tup{s, \varocircle})\}; \; \fopt \leftarrow true$ \}\;  
  }
  \lElse{ return\;}
}
  \ForAll{$A \in \sC$} {
     \If{$g_A(\gamma(s))  \not= \emptyset$} {	
 	\nl $V_2 \leftarrow V_2 \cup \{\tup{{s},A}\}$; $E \leftarrow E \cup \{(s, \tup{{s},A})\}$\;
	\nl \lIf{$g_A(\gamma(s)) \neq \gamma(s)$} {
	  $E \leftarrow E \cup \{(\tup{{s}, A}, \varotimes)\}$\;
	}
       \nl \lIf{$\fopt$} {
          $E \leftarrow E \cup \{(\tup{{s}, A}, \varocircle)\}$\;
       }
       \nl \ForAll{$s' \in g_A^{\sharp}(s)$} {
          Create a fresh array $d: up_\sC \rightarrow V_1$\;
	  \nl \ForAll{$\tup{p,c} \in up_A$}
	  {
	    $v \leftarrow \texttt{EXTRAPOLATE}(c^\sharp(s), s, A); \; d(\tup{p,c}) \leftarrow v$\;
	    \lIf{$v \not\in V_1$} { \{ $V_1 \leftarrow V_1 \cup \{v\};\; \pred(v) = \tup{s,A}; \; \text{add $v$ to $\worklist$}$ \} \;}
	  }
	    $V_p \leftarrow V_p \cup \{\tup{s,A,s',d}\}; \; E \leftarrow E \cup \{(\tup{{s},A}, \tup{s,A,s',d})\}$\;
 	}
    }
  }

\vspace{0.1cm}
{\bf Procedure} \texttt{EXTRAPOLATE}($v \in D^\sharp, s \in V_1 \setminus \{\varocircle, \varotimes\}, A \in \sC$)\;

$\tup{s', A'} \leftarrow \pred(s)$\;
\While{$\pred(s') \not= \text{nil} $} {
  \lIf{$A' = A$}{
    return $s \nabla (s \sqcup v)$\;
  }\lElse {
   \{ $\text{buffer} \leftarrow s'$; $\tup{s', A'} \leftarrow \pred(s')$; $s \leftarrow \text{buffer}$ \}\;
  }
}
return $v$\;
\caption{Computing $\sG$.}
\label{alg:compute-mdp}
\end{algorithm}

%% file: refinements.tex
\section{Refining Abstractions: Quantitative Widening Delay}
\label{sec:refinements}

Algorithm~\ref{alg:compute-g} applies the widening operator whenever 
the current node has a predecessor in the spanning tree that was created by 
the application applying the same guarded command. This strategy usually leads
to too many widenings and poor abstractions.
A popular solution in non-probabilistic abstract interpretation
is to delay widenings in an initial stage of the analysis~\cite{astree}, 
in our case until the spanning tree reaches a given depth. We call this 
approach \emph{depth-based unrolling}. Note that if $\sM_P$ is finite and the 
application of widenings is the only source of imprecision, this simple 
refinement method is complete.

A shortcoming of this approach is that it is insensitive to the probabilistic
information. We propose to combine it with another heuristic.
Given a valid abstraction $\sG$, our procedure yields two pairs 
$(\phi_1^+, \phi_2^+)$ resp. $(\phi_1^{-}, \phi_2^{-})$ 
of memoryless and non-probabilistic strategies that satisfy
$\reach{\sG[\phi_1^-, \phi_2^-], \{\varocircle\}} = \reach{\sG, \{\varocircle\}}^{+-}$ resp.
$\reach{\sG[\phi_1^+, \phi_2^+], \{\varocircle\}} = \reach{\sG, \{\varocircle\}}^{++}$.
Given  a node $s$ for Player 1, let $P_s^+$ and $P_s^-$ 
denote the probability of reaching $\varocircle$ (resp. $\varocircle$ or $\varotimes$ if we are interested
in minimal probabilities)
\andreas{Hinweis auf minimale wahrscheinlichkeiten}
 starting at $s$ and obeying
the strategies $(\phi_1^+, \phi_2^+)$ resp. $(\phi_1^-, \phi_2^-)$ in $\sG$.
In order to refine $\sG$ we can choose any node $s \in V_1 \cap D^\sharp$ 
such that $P_s^+ - P_s^- > 0$ (i.e., a node whose probability has not been
computed exactly yet), such that at least one of the direct successors of $s$
in the spanning tree has been constructed using a widening. We call
these nodes the {\em candidates} (for delaying widening). The question 
is which candidates to select. We propose to use the following simple heuristic:
\begin{quote}
Sort the candidates $s$ according to the product $w_s \cdot (P_s^+ - P_s^-)$, where 
$w_s$ denotes the product of the probabilities on the path of the spanning tree
of $\sG$ leading from $s_0$ to $s$. Choose the $n$ candidates with largest product, for 
a given $n$.
\end{quote}
We call this heuristic the {\em mass heuristic}. The {\em mixed heuristic}
delays widenings for nodes with depth less than a threshold $i$, and
for $n$ nodes of depth larger than or equal to $i$ with maximal product.
In the next section we illustrate depth-based unrolling, 
the mass heuristic, and the mixed heuristic on some examples.


%% file: experiments.tex
\subsection{Experiments}
\label{sec:experiments}

We have implemented a prototype of our approach 
on top of the Parma Polyhedra Library~\cite{parma}, which provides 
several numerical domains~\cite{grids}. 
We present some experiments showing how simple domains like intervals 
can outperform predicate abstraction. Notice that examples exhibiting the opposite behaviour 
are also easy to find: our experiments 
are not an argument against predicate abstraction, but an argument for abstraction approaches 
not limited to it.

If the computed lower and upper bounds differ by more than $0.01$, 
we select refinement candidates using the different heuristics presented before and rebuild the abstraction. 
We used a Linux machine with 4GB RAM.

\noindent{\bf Two small programs.}
Consider the NPPs of Fig.~\ref{fig:ex-series-modulo}. We compute bounds 
with different domains: intervals, octagons, integer grids, and the 
product of integer grids and intervals~\cite{cousot79}.
For the refinement we use the mass (M) depth (D) and mixed (Mix) heuristics. For M and Mix we choose
15 refinement candidates at each iteration. The results are shown in Table~\ref{tab:mod}. For the left program the integer grid domain (and the product) compute precise bounds 
after one iteration. After 10 minutes, the PASS tool~\cite{wachter} only provides the bounds
$[0.5, 0.7]$ for the optimal reachability probability. 
For the right program only the product of grids and intervals is able 
to ``see'' that $x \equiv 0 \text{ (mod 3)}$ or $y < 30$ holds,
and yields precise bounds after 3 refinement steps. After 10 minutes PASS only provides the bounds 
$[0,0.75]$. The example illustrates how pure depth-based unrolling, ignoring probabilistic information, leads to poor results: the mass and mixed heuristics perform better. 
PASS may perform better after integrating appropriate
theories, but the example shows that combining domains is powerful and easily realizable 
by using Abstract Interpretation tools.
\begin{figure}[tbp]
\begin{minipage}[t]{0.45\linewidth}
\begin{footnotesize}
\begin{verbatim}
int a=0, ctr=0;
A1: (ctr=0) 
     -> 0.5:(a'=1)&(ctr'=1) 
       +0.5:(a'=0)&(ctr'=1);
A2: (ctr=1)&(a>=-400)&(a<= 400)  
     -> 0.5:(a'=a+5)  
       +0.5:(a'=a-5);
A3: (ctr=1) -> 1:(ctr'=2);
reach: (a=1)&(ctr=2)
\end{verbatim}
\end{footnotesize}
\end{minipage}
\qquad
\begin{minipage}[t]{0.45\linewidth}
\begin{footnotesize}
\begin{verbatim}
int x=0, y=0, c=0;
A1: (c=0)&(x<=1000)
     -> 0.25:(x'=3*x+2)&(y'=y-x)          
       +0.75:(x'=3*x)&(y'=30);
A2: (c=0)&(x>1000) -> 1:(c'=1);
A3: (c=1)&(x>=3) -> 1:(x'=x-3);
reach: (c=1)&(x=2)&(y>=30)
\end{verbatim}
\end{footnotesize}
\end{minipage}
\caption{Two guarded-command programs.}\label{fig:ex-series-modulo}
\end{figure}
\begin{table}[htbp]
\iftechrep{}{\vspace{-0.5cm}}
\centering
\begin{tabular}[c]{l|l|r|r|r|r|r|r|r|r|r|r|r|r}
Program                           &  Value      & \multicolumn{3}{c|}{Interval} &  \multicolumn{3}{c|}{Octagon} &  \multicolumn{3}{c|}{Grid} & \multicolumn{3}{|c}{Product}\\ \hline
                                  &             &  M  &  D   & Mix   &  M   & D & Mix & M                          & D &  Mix & M  & D & Mix \\ \hline
\multirow{4}{*}{Left}             & Iters:      & 23  & 81   & 24    & 28      &   81  & 28       &   1  &   1  &   1     & 1  & 1 & 1    \\
                                  & Time:       & 25  & 66.1 & 27.6  & 26.6    & 63.2 & 26.9      &  $0.39$ &  $0.39$ &  $0.39$  & $0.6$ & $0.6$& $0.6$  \\  
                                  & Size:       & 793 & 667  & 769   & 681    & 691 & 681      &  17 &  17 &  17  & 61 & 61 & 61  \\ \hline
\multirow{4}{*}{Right}            & Iters:     & -& -& - & - & -     & - & -      &   -  & -          &   3   &  7  & 3      \\ 
				  & Time:       & -& -& - & - & -     & - & -      &   -  & -          &  8.3  & 20.3 & 8.2     \\ 
				  & Size:       & -& -& - & - & -     & - & -      &   -  & -          &  495  & 756 & 495     \\ 
\end{tabular}
\iftechrep{}{\vspace{0.2cm}}
\caption{Experimental results for the programs in Fig.~\ref{fig:ex-series-modulo}. Iters is the number of iterations needed.
Time is given in seconds. '-' means the analysis did not return a precise enough bound after 10 minutes. 
Size denotes the maximal number of nodes
belonging to Player 1 that occured in one of the constructed games.
}\label{tab:mod}
\end{table}

\iftechrep{}{\vspace{-0.5cm}}
\noindent {\bf Programs of Fig.~\ref{fig:prog2}.}
For these 
PASS does not terminate after 10 minutes, while with the interval domain 
our approach computes the exact value after at most 5 iterations and less than 10 seconds.  
Most of the predicates added by PASS during the refinement for program 2 have the form $\texttt{c} \leq \alpha \cdot \texttt{i} + \beta$ with $\alpha > 0, \beta < 0$:
PASS's interpolation engines seem to take the wrong guesses during the generation of new predicates. This effect remains also if we change the refinement strategies of PASS. 
PASS offers the option of manually adding predicates. Interestingly it suffices to add a predicate as simple as e.g. $i > 3$ to
help the tool deriving the solution after 3 refinements for program 2.

\begin{table}[htbp]
\centering
\begin{tabular}[c]{l|r|r|r|r|r|r}
Zeroconf protocol                 & $K = 4$  & $K = 6$    &  $K=8$  &  $K = 4$  & $K = 6$   &  $K=8$ \\
(Interval domain)                 &  P1      &   P1       &  P1     &     P2    &   P2      &  P2    \\ \hline
Time (Mass heuristic):            &  6.2     &   16.8      &  32.2    &     5.8   & 18.5       &  50.6 \\
Time (Depth heuristic):           &  2.6    &   6.0      &  6.6    &     2.6   & 6.7       &  8.1  \\
Time (Mix):                       &  2.6    &   6.3      &  6.8    &     2.6   & 6.9       &  8.4  \\
Time PASS:                        & 0.6      &  0.8       &  1.1    & 0.7       & 0.9       & 1.2    \\
\end{tabular}
\vspace{0.2cm}
\caption{Experimental results for the Zeroconf protocol. Time in seconds.}\label{tab:zeroconf}
\end{table}

\iftechrep{}{\vspace{-0.6cm}}
\noindent{\bf Zeroconf.}
This is a simple probabilistic model of the Zeroconf protocol, adapted from \cite{qproversite,kattenbelt},
where it was analyzed using PRISM and predicate abstraction. It is parameterized by $K$, 
the maximal number of probes sent by the protocol. We check it for $K=4,6,8$ and two different 
properties. 
{\it Zeroconf} is a very good example for predicate abstraction, and so it
is not surprising that PASS beats the interval domain (see Table \ref{tab:zeroconf}). The example shows 
how the mass heuristic by itself may not provide good results either, with depth-unrolling and the mixed heuristics performing substantially better.

%% file: conclusion.tex
\section{Conclusions}
\label{sec:conclusion}

We have shown that the approach of \cite{kattenbelt,probtrans} for abstraction
of probabilistic systems can be extended to arbitrary domains,
allowing probabilistic checkers to profit from well developed 
libraries for abstract domains like intervals, octagons, and polyhedra
\cite{parma,apron}. 

For this we have extended the construction of abstract reachability
trees presented in \cite{gulavani} to the probabilistic case. The 
extension no longer yields a tree, but a stochastic 2-Player game
that overapproximates the MDP semantics of the program. The correctness
proof requires to use a novel technique.

The new approach allows to refine abstractions using standard
techniques like delaying widenings.
We have also presented a technique that 
selectively delays widenings using a heuristics based 
on quantitative properties
of the abstractions.

%% file: appendix.tex
\section*{Appendix}

The key of the proof of Theorem \ref{thm:main} is the following lemma:

\begin{lemma} \mbox{}
\label{th:strategy-construction}
Let $P = (\sV, \sigma_0, \sC)$ be a NPP and $\sG$ a valid abstraction of $P$ relative to $F$ 
for the abstract domain
$D^\sharp$. Then 
\begin{itemize}
\item[(1)] Given a strategy $\phi$ of the (single) player in $\sM_P$, there exists a strategy $\phi_1 \in S_1(\sG)$ such that
 \begin{equation*}
   \inf_{\psi \in S_2(\sG)} \reach{\sG[\phi_1, \psi], \{\varocircle, \varotimes\}} \leq \reach{\sM_P[\phi], F} \leq \sup_{\psi \in S_2(\sG)} \reach{\sG[\phi_1, \psi], \{\varocircle\}}.
  \end{equation*}
\item[(2)] Given a strategy $\phi_1 \in S_1(\sG)$ there exists a strategy $\phi \in S_1(\sM_P)$ such that
\begin{equation*}
     \reach{\sM_P[\phi], F}{} \leq \sup_{\psi \in S_2(\sG)}  \reach{\sG[\phi_1, \psi], \{\varotimes, \varocircle\}}{}.
\end{equation*}
\item[(3)] Given a strategy $\phi_1 \in S_1(\sG)$ there exists a strategy $\phi \in S_1(\sM_P)$ such that
\begin{equation*}
   \reach{\sM_P[\phi], F}{} \geq \inf_{\psi \in S_2(\sG)}  \reach{\sG[\phi_1, \psi], \{\varocircle\}}{}.
\end{equation*}
\end{itemize}
\end{lemma}

\begin{proof} 
In the proof we will use some additional notation to keep the presentation short.
For a guard $g$ and a configuration $\sigma$ over $\sV$ we write
$\sigma \models g$ iff $g$ enables $\sigma$.
We will also write $\cyl{\sM[\phi], \neg F}$ as an abbreviation for the set $\cyl{\sM[\phi]} \setminus \cyl{\sM[\phi], F}$.

Let $\sG = ((V_1, V_2, V_p), E, \delta, s_0)$, and $\sM_P = ((\overline{V}_1, \overline{V}_p), \overline{E}, \overline{\delta}, \sigma_0)$.
We set $V = V_1 \cup V_2 \cup V_p$ and $\overline{V} = \overline{V}_1 \cup \overline{V}_p$.
In the following, we denote by $R, \tilde{R}, \hat{R}, \ldots$ elements in $\cyl{\sG}$,
by $r, \tilde{r}, \hat{r}, \ldots$ elements in $\cyl{\sM_P}$, by $s, \tilde{s}, \hat{s}, \ldots$ elements in $V_1$ and
by $\sigma, \tilde{\sigma}, \hat{\sigma}, \ldots$ elements in $\overline{V}_1$.

We first note that runs $R \in \cyl{\gpp}$ that end in a node in $V_1$ always can be composed in $R = R_1 R_2$, with $R_1 \in  V_1 (V_2 V_p V_1)^\ast$
and $R_2$ the empty word or $R_2 \in V_2\{\varocircle,\varotimes\}$.
Analogously, every run $r \in \cyl{\sM_P}$ ending in  $\overline{V}_1$ is contained in $\overline{V}_1 (\overline{V}_p \overline{V}_1)^\ast$. 
We will exploit this structure in the proof.

\vspace{0.3cm}
\noindent {\bf Proof of (1).} We modify $\sM_P$ as follows: We add a new node $\sigma_f$ to $\overline{V}_1$ and
edges $(q, \sigma_f)$ to $\overline{E}$ for every $q \in F$, and remove all outgoing edges
of nodes in $F$. We denote this modified MDP by $\sM$.
Note that we have $\reach{\sM[\phi], \{ \sigma_f \}} =  \reach{\sM_P[\phi], F}$ for every strategy $\phi \in S_1(\sM)$, because the only choice of $\phi$ for a history ending in  $\sigma \in F$ is $\sigma_f$ 
(and every strategy $\phi \in S_1(\sM_P)$ corresponds uniquely to
a strategy $\phi' \in S_1(\sM)$, where $\phi'$ chooses $\sigma_f$ in every node $\sigma \in F$, which is the only option there for $\phi'$).

We show that there exist strategies $\phi_1 \in S_1(\sG)$ and $\phi_2 \in S_2(\sG)$ such that
\begin{equation*}
 \reach{\mpphi,  \{ \sigma_f \}}{} = \reach{\gpp, \{\varocircle\}}{},
\end{equation*}
which implies 
$\reach{\gpp, \{\varocircle\}}{} \geq \reach{\sM_P[\phi], \{ \sigma_f \}}$.
Since $\phi_2$ in our proof never chooses $\varotimes$ we also get
$\reach{\gpp, \{\varocircle, \varotimes\}}{} \leq \reach{\sM_P[\phi], \{ \sigma_f \}}$.

The crucial point of the proof is to distribute the probabilistic mass of a run $R \in \cyl{\sG}$
(ending in a node of the first player) among all the runs
$r \in \cyl{\sM}$ in a suitable way, depending on the strategy $\phi$. This distribution is then used 
to define the strategies $\phi_1$ and $\phi_2$. We formalize the distribution as a function
\begin{equation*}
 D: (\cyl{\sG}\cap V^\ast V_1) \times (\cyl{\sM}\cap \overline{V}^\ast \overline{V}_1) \rightarrow [0,1].
\end{equation*}
\andreas{Andreas: I've replaced $\gpp$ by $\sG$!!}
\noindent where, loosely speaking, $D(R,r)$ is the fraction of the probabilistic mass of $R$ that is 
assigned to $r$. In order to define $D$ we need an auxiliary function $\beta$. Due to the requirements of a valid abstraction, we know that for every $\tup{s,A} \in V_2$ and every
concrete configuration $\sigma \in \gamma(g_A(s))$ there exists at least one node $\tup{s,A,s',d} \in V_p$ with $\tup{s,A}\rightarrow\tup{s,A,s',d}$ and $\sigma \in \gamma(s')$. But there might
be more than one node in $V_p$ satisfying these conditions. We fix for every $\tup{s,A}$ and every $\sigma \in \gamma(g_a(s))$  an arbitrary successor node $\tup{s,A,s',d}$ with $\sigma \in \gamma(s')$ and set
$\beta(\sigma, s, A) = \tup{s,A,s',d}$.

Now we finally proceed to define $D$ inductively:
\begin{itemize}
\item[(i)] If $R = s_0$ and $r = \sigma_0$ then $D(R,r) := 1$.
\item[(ii)]  If $R = \tilde{R}\tup{s,A}\tup{s,A,s',d}\tilde{s}$, $r = \tilde{r}\tup{\tilde{\sigma}, A}\sigma$, 
$A \in \sC$ and $\beta(\tilde{\sigma}, s, A) = \tup{s,A,s',d}$ then 
\begin{equation*}
 D(R,r) := \sum_{\substack{\tup{p,c} \in up_A: \\ 
 d(\tup{p,c}) = \tilde{s} \wedge \sem{c}(\tilde{\sigma}) = \sigma}}
 {\phi(\tilde{r})(\tup{\tilde{\sigma}, A}) \cdot D(\tilde{R}, \tilde{r}) \cdot p}.
 \end{equation*} 
\item[(iii)] If $R = \tilde{R}\tilde{s}\tup{\tilde{s},\varocircle}\varocircle$ and $r = \tilde{r}\sigma\sigma_f $ then
$D(R,r) := D(\tilde{R}\tilde{s}, \tilde{r}\sigma)$.
\item[(iv)] Otherwise $D(R,r) := 0$.
\end{itemize}
\noindent Note that, due to the properties of a valid abstraction and our definition of $D$, if $D(Rs, r\sigma) > 0$, then $\sigma \in \gamma(s)$ for every $\sigma \in \Sigma_\sV$.
Also note that for every $\tilde{r} \in \cyl{\sM}$ 
We now proceed to define the strategies $\phi_1, \phi_2$. We use the abbreviations 
$$\begin{array}{rcl}
D_R  :=  \sum_{r \in \cyl{\mpphi}}{D(R, r)} & \quad \mbox{ and } \quad &
D_r  :=  \sum_{R \in \cyl{\gpp}}{D(R,r)}
\end{array}$$
Given a run $R \in \cyl{\sG}\cap V^\ast V_1$, if $D_R = 0$ we define $\phi_1(R)$ arbitrarily, i.e., we let $\phi_1(R)$
be an arbitrary successor of the last node of $R$. If $D_R \neq 0$, we define $\phi_1(R)$ as follows:

\begin{align*}
  \phi_1(R)(\tup{s,A}) &:= \frac{1}{D_R} \cdot \sum_{r \in \cyl{\sM[\phi], \neg F}}{D(R, r) \cdot \phi(r)(\tup{\sigma,A})} \\
  \phi_1(R)(\tup{s,\varocircle}) &:= \frac{1}{D_R} \cdot \sum_{r \in \cyl{\mpphi, F}}{D(R, r)}
\end{align*}
\noindent Given a run $R = \tilde{R}\tup{s,A} \in \cyl{\sG}\cap \overline{V}^\ast \overline{V}_2$ and $A \in \sC$, if $D_{\tilde{R}} = 0$ we define $\phi_2(R)$ arbitrarily. If $D_{\tilde{R}} \neq 0$, 
we define $\phi_2(R)$ as follows:
\begin{align*}
  \phi_2(\tilde{R}\tup{s,A})(\tup{s,A,s',d}) &:= 
    \frac{ \displaystyle{\sum_{\substack{r  \in \cyl{\sM[\phi], \neg F}: \\ \beta(\sigma, s, A) = \tup{s,A,s',d} \\ }}{D(\tilde{R}, r) \cdot \phi(r)(\tup{\sigma,A})}} }{D_{\tilde{R}} \cdot \phi_1(\tilde{R})(\tup{s,A})} \\
  \phi_2(\tilde{R}\tup{s,\varocircle})(\varocircle) := 1
\end{align*}

\noindent It is easy to see that the functions $\phi_1, \phi_2$ so defined are
indeed strategies: Recall that a NPP cannot reach a configuration $\sigma \not \in F$ where no guarded command is enabled. Hence $\phi(r)$ is always defined when used in the definition.


We list several properties of the function $D$:
\begin{itemize}
\item[(a)] For every $r \in \cyl{\mpphi, \{\sigma_f\}}$: if $D(R,r) > 0$, then $R \in \cyl{\gpp, \{\varocircle\}}$.\\
(Follows easily from the definitions.)
\item[(b)] For every $R \in \cyl{\gpp, \{\varocircle\}}$: if $D(R,r) > 0$, then $r \in \cyl{\mpphi, \{\sigma_f\}}$ (Follows easily from the definitions.)
\item[(c)] For every $R \in \cyl{\gpp}$ ending in $V_1 \cup \{\varocircle, \varotimes\}$: $D_R = \val{R}{}$.\\
(Proof delayed, see below. Note that due to our
choice of $\phi_1, \phi_2$, $\varotimes$ can never be reached.)
\item[(d)] For every $r \in \cyl{\gpp}$ ending in $\overline{V}_1 \cup \{\sigma_f\}$: $D_r = \val{r}{}$.\\
(Proof delayed, see below.)
\end{itemize}

Using (a)-(d) we can now conclude the proof:
\begin{align*}
\reach{\gpp, \{\varocircle\}} &= \sum_{R \in \cyl{\gpp, \{\varocircle\}}}{\val{R}{}} & \text{(Def. of $\reach{}$)}\\
&= \sum_{R \in \cyl{\gpp, \{\varocircle\}}}{\sum_{r \in \cyl{\mpphi}}{D(R, r)}} & \text{(Prop. (c))} \\
&= \sum_{r \in \cyl{\mpphi, \{\sigma_f\}}} \sum_{R \in \cyl{\gpp}}{D(R,r)} & \text{(Props. (a)-(b))}\\
&= \sum_{r \in \cyl{\mpphi, \{\sigma_f\}}} \val{r}{} = \reach{\mpphi} & \text{(Prop. (d))}.
\end{align*}

\paragraph{Proof of property (c).} By induction on the length of $R$.
\begin{itemize}
 \item If $R = s_0$ then $\sum_{\substack{r \in \cyl{\mpphi}}}{D(s_0, r)} = D(s_0, \sigma_0) = 1 = \val{R}{}_{\gpp}$.
 \item If $R = \tilde{R} \tup{s,A}\tup{s,A,s',d}\tilde{s} \in \cyl{\gpp}$ then
$$
\begin{array}{rclr}
& &\val{R}{} \\[0.2cm]
&=&\val{\tilde{R}}{} \cdot  \phi_1(\tilde{R})(\tup{s,A}) \cdot \phi_2(\tilde{R} \tup{s,A}) \cdot  \delta(\tup{s,A,s',d})(\tilde{s}) & \text{(Def. $\val{\cdot}{}$)} \\[0.2cm]
&=& D_{\tilde{R}} \cdot \phi_1(\tilde{R})(\tup{s,A}) \cdot \phi_2(\tilde{R} \tup{s,A})(\tup{s,A,s',d}) \cdot  \delta(\tup{s,A,s',d})(\tilde{s}) & \text{(Ind. Hyp.)}\\[0.2cm]
&=& \displaystyle
\Bigg(  \sum_{\substack{\tilde{r}\tup{\sigma,A} \in \cyl{\sM[\phi], \neg F}: \\ \beta(\sigma, s, A) = \tup{s,A,s',d} \\ }}
{D(\tilde{R}, \tilde{r})} \cdot \phi(\tilde{r})(\tup{\sigma,A})\Bigg) \cdot \delta(\tup{s,A,s',d})(\tilde{s}) & \text{(Def. $\phi_1, \phi_2$)}
\\[0.2cm]
&=& \displaystyle
\sum_{\substack{\tup{p,c} \in up_A: \\ d(\tup{p,c}) = \tilde{s}}}
\Bigg(  \sum_{\substack{\tilde{r}\tup{\sigma,A}\in \cyl{\sM[\phi], \neg F}: \\ \beta(\sigma, s, A) = \tup{s,A,s',d} \\
 }}{p \cdot D(\tilde{R}, \tilde{r}) \cdot \phi(\tilde{r})(\tup{\sigma,A})}\Bigg)& \text{(Def. $\delta$)} 
\\[0.2cm]
&=& \displaystyle
\sum_{\substack{\tup{p,c} \in up_A: \\ d(\tup{p,c}) = \tilde{s}}}
\Bigg(  \sum_{\substack{\tilde{r}\tup{\sigma,A}\tilde{\sigma} \in \cyl{\mpphi}: \\ \beta(\sigma, s, A) = \tup{s,A,s',d} \\ \wedge \sem{c}(\sigma) = \tilde{\sigma}}}{p \cdot D(\tilde{R}, \tilde{r}) \cdot \phi(\tilde{r})(\tup{\sigma,A})}\Bigg) & \text{(**)}
\\[0.2cm]
&=& \displaystyle
\sum_{\substack{\tilde{r}\tup{\sigma,A}\tilde{\sigma}  \in \cyl{\sM[\phi], \neg F}}}
 \sum_{\substack{\tup{p,c} \in up_A: \\ d(\tup{p,c}) = \tilde{s} \wedge \sem{c}(\sigma) = \tilde{\sigma} \\  \wedge \beta(\sigma, s, A) = \tup{s,A,s',d} \\
  }}
{p \cdot D(\tilde{R}, \tilde{r}) \cdot \phi(\tilde{r})(\tup{\sigma,A})} 
\\[0.2cm]
&=&
D_R. & \text{ (Def. $D$)}.
\end{array}$$
\noindent (**) For every run $\tilde{r}\tup{\sigma,A}$ and every $\tup{p,c} \in \sC$ there exists
exactly one extension $\tilde{r}\tup{\sigma,A}\sem{c}(\sigma)$, and so the value of the sum does not change by modifying the sum quantifier in this way.
\item If $R = \tilde{R}\tilde{s}\tup{\tilde{s}, \varocircle}\varocircle \in \cyl{\gpp}$ then
 \begin{align*}
   \val{R}{} &= \val{\tilde{R}\tilde{s}}{} \cdot\phi_1(\tilde{R}\tilde{s})(\tup{\tilde{s}, \varocircle}) \cdot \phi_2(\tilde{R}\tilde{s}\tup{\tilde{s}, \varocircle})(\varocircle) \\
   &= D_{\tilde{R}\tilde{s}} \cdot\phi_1(\tilde{R}\tilde{s})(\tup{\tilde{s}, \varocircle}) \cdot \phi_2(\tilde{R}\tilde{s}\tup{\tilde{s}, \varocircle})(\varocircle) & \text{ (Ind. hyp.)} \\
   &= D_{\tilde{R}\tilde{s}} \cdot\phi_1(\tilde{R}\tilde{s})(\tup{\tilde{s}, \varocircle}) & \text{($\phi_2(\tilde{R}\tilde{s}\tup{\tilde{s}, \varocircle})(\varocircle)=1$ by def. of  $\phi_2$)} \\
   &= \sum_{r \in \cyl{\mpphi, F}}{D(\tilde{R}\tilde{s}, r)}  & \text{ (Def. of $D_R$)}\\ 
   &= \sum_{r \in \cyl{\mpphi, F}}{D(\tilde{R}\tilde{s}\tup{\tilde{s}, \varocircle}\varocircle, r\sigma_f)}  & \text{ (Def. of $D$, part (iii))}\\ 
   &= D_R. &  \text{ (Def. of $D$, Props. (a)-(b))} 
 \end{align*}
\end{itemize}

\noindent {\it Proof of property (d).} By induction on the length of $r$.
\begin{itemize}
\item If $r = \sigma_0$ then we proceed as in the case of $R$.
\item If  $r = \tilde{r} \tup{\tilde{\sigma}, A} \sigma \in \cyl{\mpphi}$ then
$$\begin{array}{rclr}
  & & \val{r}{} \\[0.2cm]

  &=& \displaystyle \val{\tilde{r}}{}_{} \cdot \phi(\tilde{r})(\tup{\tilde{\sigma}, A}) \cdot 
\sum_{\tup{p,c} \in up_A: \sem{c}(\tilde{\sigma}) = \sigma}{p} \\ [0.5cm]

  &=&  \displaystyle \Bigg(  \sum_{\tilde{R}s \in \cyl{\gpp}}{D(\tilde{R}s, \tilde{r})} \Bigg) \cdot \phi(\tilde{r})(\tup{\tilde{\sigma}, A}) 
  \cdot \sum_{\tup{p,c} \in up_A: \sem{c}(\tilde{\sigma}) = \sigma}{p} & \text{ (Ind. Hyp.) } \\[0.5cm]

  &=&\displaystyle \sum_{\tilde{R}s \in \cyl{\gpp}} \sum_{\substack{\tup{p,c} \in up_A: \\ \sem{c}(\tilde{\sigma}) = \sigma}}
  {D(\tilde{R}s, \tilde{r})} \cdot \phi(\tilde{r})(\tup{\tilde{\sigma}, A}) \cdot p & \text{ } \\[0.5cm]

  &=&\displaystyle 
\sum_{\substack{\tilde{R}s\tup{s,A}\tup{s,A,s',d}\\ \in \cyl{\gpp}: \\  \beta(\tilde{\sigma}, s,A) = \tup{s,A,s',d}}}
\sum_{\substack{\tup{p,c} \in up_A: \\ \sem{c}(\tilde{\sigma}) = \sigma }}
  D(\tilde{R}s, \tilde{r}) \cdot \phi(\tilde{r})(\tup{\tilde{\sigma}, A}) \cdot p & \text{(**)} \\[0.2cm]

  &=&\displaystyle 
\sum_{\substack{\tilde{R}s\tup{s,A}\tup{s,A,s',d}\tilde{s}\\ \in \cyl{\gpp}: \\ \beta(\tilde{\sigma}, s,A) = \tup{s,A,s',d}}}
\sum_{\substack{\tup{p,c} \in up_A: \\ \sem{c}(\tilde{\sigma}) = \sigma \\ \wedge d(\tup{p,c}) = \tilde{s}}}
  D(\tilde{R}s, \tilde{r}) \cdot \phi(\tilde{r})(\tup{\tilde{\sigma}, A}) \cdot p & \text{(***)} \\[0.2cm]

  &=&\displaystyle 
\sum_{\substack{\tilde{R}s\tup{s,A}\tup{s,A,s',d}\tilde{s} \\ \in \cyl{\gpp}}}
\sum_{\substack{\tup{p,c} \in up_A: \\ d(\tup{p,c}) = \tilde{s} \wedge \sem{c}(\tilde{\sigma}) = \sigma \\ \wedge \beta(\tilde{\sigma}, s,A) = \tup{s,A,s',d}}}
  D(\tilde{R}s, \tilde{r}) \cdot \phi(\tilde{r})(\tup{\tilde{\sigma}, A}) \cdot p &  \\[0.2cm]

  &=& \displaystyle D_r.

\end{array}$$
\noindent (**)
If $D(\tilde{R}s, \tilde{r}) > 0$, then $\tilde{\sigma} \in \gamma(s)$ holds. 
Using the definition of a valid abstraction we conclude that $s \rightarrow \tup{s,A}$ if $\tilde{\sigma} \rightarrow \tup{\tilde{\sigma}, A}$ 
(especially if $\phi(\tilde{r})(\tup{\tilde{\sigma}, A}) > 0$).
Also there exists a unique node $\tup{s, A, s', d}$ with 
$\tup{s, A} \rightarrow \tup{s,A,s',d}$ and $\tup{s, A, s', d} = \beta(\tilde{\sigma}, A, \hat{s})$.
Therefore we can replace the sum quantifier without changing the summands here.

\noindent (***) For every $t = \tup{p,c} \in up_a$, there exists a unique $\tilde{s} \in V_1$ such that
$d(t) = \tilde{s}$ and $\tup{s,A,s',d} \rightarrow \tilde{s}$. 
So, for a fixed $\tilde{R}s\tup{s,A}\tup{s,A,s',d}$ with
$\beta(\tilde{\sigma}, s,A) = \tup{s,A,s',d}$, 
\begin{equation*}
 \sum_{\substack{\tup{p,c} \in up_A: \\ \sem{c}(\tilde{\sigma}) = \sigma }}
  D(\tilde{R}s, \tilde{r}) \cdot \phi(\tilde{r})(\tup{\tilde{\sigma}, A}) \cdot p 
= 
 \sum_{\substack{\tilde{s} \in V_1: \\ \tilde{R}s\tup{s,A}\tup{s,A,s',d}\tilde{s} \\ \in \cyl{\gpp} }}
 \sum_{\substack{\tup{p,c} \in up_A: \\ \sem{c}(\tilde{\sigma}) = \sigma \\ \wedge d(\tup{p,c}) = \tilde{s}}}
  D(\tilde{R}s, \tilde{r}) \cdot \phi(\tilde{r})(\tup{\tilde{\sigma}, A}) \cdot p 
\end{equation*}

\item If $r = \tilde{r}\tilde{\sigma}\sigma_f$ then
\begin{align*}
\val{r}{} &= \val{\tilde{r}\tilde{\sigma}}{} \cdot \phi(\tilde{r}\tilde{\sigma})(\sigma_f)  \\[0.2cm]
&= \val{\tilde{r}\tilde{\sigma}}{} & \text{(**)}  \\[0.2cm]
&= D_{\tilde{r}\tilde{\sigma}} & \text{ (Ind. hyp.) }  \\
&= D_r & \text{ (Def. of $D$). } &
\end{align*}
(**) By definition of $\sM$, the only successor of $\sigma$ is $\sigma_f$.
\end{itemize}

\vspace{0.3cm}
\noindent {\bf Proof of (2).} Recall that at the beginning of the proof of (1) we have extended
$\sM_P$ by adding a new node $\sigma_f$ to 
$\overline{V}_1$ and by adding edges $(q, \sigma_f)$ to $\overline{E}$ for every $q \in F$.
We now extend $\sM_P$ further by also adding edges $(q, \sigma_f)$ to $\overline{E}$ for every $q \in \overline{V}_1 \setminus F$.
Now \emph{every} node in $V_1 \setminus \{\sigma_f\}$ has $\sigma_f$ as a successor. We still call this modified MDP $\sM$.

Let $\phi_1 \in S_1(\sG)$ be a strategy for Player 1. We show that there exist strategies 
$\phi_2 \in S_2(\sG)$ and $\phi \in S_1(\sM)$ such that
\begin{equation*}
\label{eq:m-eq}
 \reach{\sM[\phi], \{\sigma_f\}}{} = \reach{\gpp, \{\varocircle, \varotimes\}}{}.
\end{equation*}

To show that this equality proves (2), observe that for every $\phi \in S_1(\sM)$ there exists a strategy $\phi' \in S_1(\sM_P)$ such that
$\reach{\sM_P[\phi'], F} \leq \reach{\sM[\phi], \{\sigma_f\}}$ (simply distribute the probability assigned to
$\sigma_f$ to other successors arbitrarily if necessary). 
So for every $\phi \in S_1(\sM)$ that satisfies the equality we have
\begin{equation*}
\reach{\sM[\phi'], F} \leq \reach{\sG[\phi_1, \phi_2], \{\varocircle, \varotimes\}}
\leq  \sup_{\psi \in S_2(\sG)}  \reach{\sG[\phi_1, \psi], \{\varotimes, \varocircle\}}{}.
\end{equation*}

We now observe that it suffices to prove the results for strategies $\phi_1$ 
that are memoryless and non-probabilistic: this follows easily from Theorem~\ref{th:strategy-construction},
which shows that the infimum over all $\phi_1$ of $\sup_{\psi \in S_2(\sG)} \reach{\sG[\phi_1, \psi], \{\varotimes, \varocircle\}}{}$ is achieved for such a strategy $\phi_1$.
Hence as an abbreviation we can write $\phi_1(s) = A$ for $\phi_1(Rs)(\tup{s,A}) = 1$ for all 
$s \in V_1$, $Rs \in \cyl{\gpp}$ and $a \in \sC \cup \{\varotimes, \varocircle\}$.

For the proof we again define a suitable distribution function $D$, this time as follows:
\begin{itemize}
 \item  If $R = \sigma_0$ and $r = \sigma_0$ then $D(R,r) := 1$.
 \item If $R = \tilde{R}\tup{\tilde{s},A}\tup{\tilde{s},A,s',d}s$ and
  $r = \tilde{r}\tup{\tilde{\sigma}, A}\sigma$, with
  $\tup{\tilde{s},A} \not\rightarrow \varotimes$ and $\tup{\tilde{s},A} \not\rightarrow \varocircle$ and
  $\beta(\sigma, \tilde{s},A) = \tup{\tilde{s},A,s',d}$ then
  \begin{equation*}
    D(R,r) := D(\tilde{R}, \tilde{r}) \cdot \sum_{\substack{\tup{p,c}\in up_A:  \\ d(\tup{p,c}) = s  \wedge{\sem{c}(\tilde{\sigma}}) = \sigma}}{p}.
  \end{equation*}
 \item If $R = \tilde{R}\tup{\tilde{s},A}\varocircle$ 
 with  $\tup{\tilde{s}, A} \rightarrow \varocircle$  (with $A\in \sC \cup \{\varocircle\}$)
 and $r = \tilde{r} \sigma \sigma_f$ then
  \begin{equation*}
    D(R,r) := D(\tilde{R}, \tilde{r}\sigma).
  \end{equation*}
\item If $R = \tilde{R}\tup{\tilde{s},A}\varotimes$
 with  $\tup{\tilde{s}, A} \rightarrow \varotimes$ and $\tup{\tilde{s}, A} \not\rightarrow \varocircle$ (with $A\in \sC \cup \{\varocircle\}$)
 and $r = \tilde{r} \sigma \sigma_f$ then
  \begin{equation*}
    D(R,r) := D(\tilde{R}, \tilde{r}\sigma).
  \end{equation*}
 \item Otherwise $D(R,r) := 0$.
\end{itemize}

We now proceed to define $\phi$. Let ${\it final}(\tup{s,A})$ denote that $\tup{s,A} \rightarrow \varocircle$ or $\tup{s,A} \rightarrow \varotimes$. Let $r\sigma \in \cyl{\sM}$. 
\begin{itemize}
\item If $D_{r\sigma} = 0$ then we define $\phi(r)$ arbitrarily. 
\item If $D_{r\sigma} > 0$ then for every $A \in \sC$
\begin{equation*}
 \phi(r\sigma)(\tup{\sigma, A}) := \frac{1}{D_{r\sigma}} \sum_{\substack{R s\in \cyl{\gpp}: \\ \phi_1(s) = A  \wedge \neg {\it final}(\tup{s,A})}}{D(R s, r\sigma)}
\end{equation*}
and
\begin{equation*}
 \phi(r\sigma)(\sigma_f) := \frac{1}{D_{r\sigma}} \sum_{\substack{Rs \in \cyl{\gpp}: \\ \phi_1(s) = A  \wedge  {\it final}(\tup{s,A})}}{D(R s, r\sigma)}
\end{equation*}

Note here that this is a valid strategy: If $\sigma \not\rightarrow \tup{\sigma, A}$, then we know 
(due to the definition of a valid abstraction
and again the fact that $\sigma \in \gamma(s)$ if $D(\tilde{R}s, \tilde{r}\sigma)>0$) that $\tup{s,a} \rightarrow \varotimes$, 
hence $\phi(r\sigma)(\tup{\sigma, A}) = 0$. 
Also it is easy to see that $\phi(r\sigma)(\sigma_f) +  \sum_{A \in \sC}{\phi(r\sigma)(\tup{\sigma,A})} = 1$.
\end{itemize}

Finally, we define $\phi_2 \in S_2(\sG)$. As in part (1), $\phi_2$ directs the probabilistic mass to the matching nodes in $V_p$ relative
to $\beta$. But in this case we choose $\varocircle$ resp. $\varotimes$ as often as possible:
For $R = \tilde{R}s\tup{s,a} \in \cyl{\gpp}$ (i.e., $\phi_1(s) = A$),
\begin{itemize}
 \item If $\tup{s,A} \rightarrow \varocircle$ then $\phi_2(R)(\varocircle)=1$.
 \item If $\tup{s,A} \not\rightarrow \varocircle$ and $\tup{s,A} \rightarrow \varotimes$ then $\phi_2(R)(\varotimes)=1$.
 \item If $\tup{s,A} \not\rightarrow \varocircle$, $\tup{s,A} \not\rightarrow \varotimes$, and $D_R = 0$ then we define $\phi_2(R)$ arbitrarily.
 \item If $\tup{s,A} \not\rightarrow \varocircle$, $\tup{s,A} \not\rightarrow \varotimes$, and $D_R > 0$ then for every $\tup{s,A,s',d}$ with $\tup{s,A} \rightarrow \tup{s,A,s',d}$
 \begin{equation*}
  \phi_2(R)(\tup{s,A,s',d}) = \frac{1}{D_R} \cdot \sum_{\substack{r\sigma \in \cyl{\sM}: \\ \beta(\sigma,s,A) = \tup{s,A,s',d}}} {D(\tilde{R}s, r\sigma)}.
 \end{equation*}
\end{itemize}

Now, as in the proof of (1), we show that with this choice for $D$, $\phi$, and $\phi_2$ the
properties (a)-(d) hold. The rest of the proof is then exactly as in the proof of (1). Proving properties (a) and (b) is again easy. 

\vspace{0.2cm}
\noindent {\it Proof of property (c).} By induction on the length of $R$.
\begin{itemize}
 \item If $R = s_0$ then $\val{R}{} = 1 = D(R, \sigma_0) = D_R$.
 \item For $R = \tilde{R}\tilde{s}\tup{\tilde{s}, A}\tup{\tilde{s}, A, s', d}s$ then the proof proceeds as for (1). 
 \item For $R = \tilde{R}\tilde{s}\tup{\tilde{s}, A}\varocircle$, $\phi_1(\tilde{s}) = A$ and $\tup{\tilde{s}, A} \rightarrow \varocircle$:
 \begin{align*}
   \val{R}{} &= \val{\tilde{R}\tilde{s}}{} \cdot \phi_2(\tilde{R}\tilde{s}\tup{\tilde{s}, A})(\varocircle) 
  & \text{}\\
   &= \val{\tilde{R}}{} & \text{$\phi_2(\tilde{R}\tilde{s}\tup{\tilde{s}, A})(\varocircle)=1$ by Def. of $\phi_2$}\\
   &= D_{\tilde{R}} & \text{ (Ind. hyp.) } \\
   &= D_R. &  \text{ (Def. of $D_R$ and Def. of $D$, part (iii))}
 \end{align*}
  If  $\tup{\tilde{s}, A} \not\rightarrow \varocircle$, then $\val{R}{} = \val{\tilde{R}\tilde{s}}{} \cdot \phi_2(\tilde{R}\tilde{s}\tup{\tilde{s}, A})(\varocircle) = 0 = D_R$. 
 \item For $R = \tilde{R}\tilde{s}\tup{\tilde{s}, A}\varotimes$, $\phi_1(\tilde{s}) = A$, $\tup{\tilde{s}, A} \not\rightarrow \varocircle$
and $\tup{\tilde{s}, A} \rightarrow \varotimes$:
 \begin{align*}
   \val{R}{} &= \val{\tilde{R}\tilde{s}}{} \cdot \phi_2(\tilde{R}\tilde{s}\tup{\tilde{s}, A})(\varotimes) 
  & \text{}\\
   &= \val{\tilde{R}\tilde{s}}{} & \text{$\phi_2(\tilde{R}\tilde{s}\tup{\tilde{s}, A})(\varotimes)=1$ by Def. of $\phi_2$}\\
   &= D_{\tilde{R}\tilde{s}} & \text{ (Ind. hyp.) } \\
   &= D_R. &  \text{ (Def. of $D_R$ and Def. of $D$, part (iv))  }
 \end{align*}
  If  $\tup{\tilde{s}, A} \rightarrow \varocircle$ or  $\tup{\tilde{s}, A} \not\rightarrow \varotimes$ holds,
 then $\val{R}{} = \val{\tilde{R}\tilde{s}}{} \cdot \phi_2(\tilde{R}\tilde{s}\tup{\tilde{s}, A})(\varotimes) = 0 = D_R$. 
\end{itemize}

\noindent {\it Proof of property (d).} By induction on the length of $r$.

\begin{itemize}
 \item If $r = \sigma_0$ then $\val{r}{} = 1 = D(s_0, r) = D_r$.
 \item If $r = \tilde{r}\tilde{\sigma}\tup{\sigma,A}\sigma$ then
$$\begin{array}{rclr}
 & & \val{r}{} \\[0.2cm]
 &=& \displaystyle \val{\tilde{r}\tilde{\sigma}}{} \cdot \phi(\tilde{r}\tilde{\sigma})(\tup{\tilde{\sigma}, A}) \cdot 
 \sum_{\substack{\tup{p,c}\in up_a: \\{\sem{c}(\tilde{\sigma}) = \sigma}}}{p} 
& \text{(Def. $\val{\cdot}{}$)}\\[0.9cm]
 &=& \displaystyle D_{\tilde{r}\tilde{\sigma}} \cdot \frac{1}{D_{\tilde{r}\tilde{\sigma}}} \cdot \sum_{\substack{\tilde{R}s \in \cyl{\gpp}: \\ \phi_1(s) = A \wedge
  \neg final(\tup{s,A})}}{D(\tilde{R}s, \tilde{r}\tilde{\sigma})}
 \cdot 
 \sum_{\substack{\tup{p,c}\in up_A: \\ {\sem{c}(\tilde{\sigma}) = \sigma}}}{p}
 & \text{ (Ind. hyp., Def. $\phi$) }  \\[0.9cm]
 &=& \displaystyle \sum_{\substack{\tilde{R}s \in \cyl{\gpp}: \\ \phi_1(s) = A \wedge \neg final(\tup{s,A})}}{D(\tilde{R}s, \tilde{r}\tilde{\sigma})}
 \cdot 
 \sum_{\substack{\tup{p,c}\in up_a: \\ {\sem{c}(\tilde{\sigma}) = \sigma}}}{p}
 & \text{ } \\[0.9cm]
 &=& \displaystyle \sum_{\substack{\tilde{R}s\tup{s,A,s',d} \in \cyl{\gpp}: \\ \phi_1(s) = A  \wedge \neg final(\tup{s,A}) \\ \wedge \beta(\tilde{\sigma},s,A) = \tup{s,A,s',d}}}{D(\tilde{R}s, \tilde{r}\tilde{\sigma})} 
 \cdot 
 \sum_{\substack{\tup{p,c}\in up_a: \\ {\sem{c}(\tilde{\sigma}) = \sigma} }}{p}
 &\text{(**)} \\[0.9cm]
 &=& \displaystyle \sum_{\substack{\tilde{R}s\tup{s,A,s',d} \in \cyl{\gpp}: \\ \phi_1(s) = A  \wedge \neg final(\tup{s,A}) \\ \wedge \beta(\tilde{\sigma},s,A) = \tup{s,A,s',d}}}
  \sum_{\substack{\tup{p,c}\in up_a: \\ {\sem{c}(\tilde{\sigma}) = \sigma} }} {p \cdot D(\tilde{R}s, \tilde{r}\tilde{\sigma})} 
 &\text{} \\[0.9cm]
 &=& \displaystyle \sum_{\substack{\tilde{R}s\tup{s,A,s',d}\tilde{s} \in \cyl{\gpp}: \\ \phi_1(s) = A  \wedge \neg final(\tup{s,A}) \\ \wedge \beta(\tilde{\sigma},s,A) = \tup{s,A,s',d}}}
  \sum_{\substack{\tup{p,c}\in up_A: \\ {\sem{c}(\tilde{\sigma}) = \sigma} \\ \wedge d(\tup{p,c}) = \tilde{s} }} {p \cdot D(\tilde{R}s, \tilde{r}\tilde{\sigma})} 
 &\text{(see explanation in part (1))} \\[0.9cm]
  &=& D_r & \text{(Def. of $D$).}
\end{array}$$
\noindent(**) As in part (1), we use that there exists a unique $\tup{s,A,s',d}$ with $\beta(\tilde{\sigma}, s,A) = \tup{s,A,s',d}$.

\item If $r = \tilde{r}\tilde{\sigma}\sigma_f$ then 
$$\begin{array}{rclr}
& & \val{r}{} \\[0.2cm]
&=& \val{\tilde{r}\tilde{\sigma}}{} \cdot \phi(\tilde{r}\tilde{\sigma})(\sigma_f)  
& \text{(Def. $\val{\cdot}{}$)} \\
&=& \displaystyle D_{\tilde{r}\tilde{\sigma}} \cdot \frac{1}{D_{\tilde{r}\tilde{\sigma}}} \cdot \sum_{\substack{R \in \cyl{\gpp}: \\ \phi_1(s) = A \wedge  final(\tup{s,a})}}{D(R, r)}
& \text{(Ind. hyp., Def. $\phi$)}  \\[0.9cm]
&=& \displaystyle \sum_{\substack{R \in \cyl{\gpp}: \\ \phi_1(s) = A \wedge  final(\tup{s,a})}}{D(R, r)} \\[0.9cm]
&=& D_r & \text{(Def. of $D_r$ and $D$)}
\end{array}$$
\end{itemize}
\noindent {\bf Proof of (3).}
This proof is similar to the one of part (2), but now Player 2 always chooses $\varotimes$ if she can.
Therefore 
we add an additional node $\sigma_{\varotimes}$ to $\sM$ which is no goal state, and connect every node in $\overline{V}_1$ with $\sigma_{\varotimes}$. 
These edges now simulate the choice of $\varotimes$ as successor node in $\sG[\phi_1, \phi_2]$. The rest of the proof is a simple variation of the one of part (2).
\end{proof}

We can now proceed to proving the theorem:

\setcounter{theorem}{1}
\begin{theorem}
 Let $P$ be a NPP and let $\sG$ be a valid abstraction of $P$ relative to $F$ for the abstract domain
$D^\sharp$. Then 
\begin{align*}
 \reach{\sM_P, F}^{-} &\in [\reach{\sG, \{\varocircle, \varotimes\}}^{--}, \reach{\sG, \{\varocircle, \varotimes\}}^{-+}] \text{ and } \\
 \reach{\sM_P, F}^{+} &\in [\reach{\sG, \{\varocircle\}}^{+-}, \reach{\sG, \{\varocircle\}}^{++}].
\end{align*}
\end{theorem}

\begin{proof} Let $\phi \in S_1(\sM)$ be a strategy of the (single) player in $\sM_P$. By
Lemma~\ref{th:strategy-construction}(1)  there exists a 
strategy $\phi_1 \in S_1(\sG)$ satisfying
\begin{equation*}
   \inf_{\psi \in S_2(\sG)} \reach{\sG[\phi_1, \psi], \{\varocircle, \varotimes\}}{} \leq \reach{\sM_P[\phi], F}{}. 
\end{equation*}
From this we conclude that for all $\phi \in S_1(\sM_P)$
\begin{equation*}
 \reach{\sG, \{\varocircle, \varotimes\}}^{--} \leq \inf_{\psi \in S_2(\sG)} \reach{\sG[\phi_1, \psi], \{\varocircle, \varotimes\}}{} \leq \reach{\sM_P[\phi], F}{}
\end{equation*}
and hence 
\begin{equation*}
\reach{\sG, \{\varocircle, \varotimes\}}^{--} \leq \inf_{\phi \in S_1(\sM_P)} \reach{\sM_P[\phi], F}{} = \reach{\sM_P, F}^{-}.           
\end{equation*}
The inequality $\reach{\sM_P, F}^{+}  \leq \reach{\sG, \{\varocircle\}}^{++}$
can be proved in the same way by using the left inequality of Lemma~\ref{th:strategy-construction}(1). The remaining 
inequations $\reach{\sM_P, F}^{-}  \leq \reach{\sG, \{\varocircle, \varotimes\}}^{-+}$ and
$\reach{\sM_P, F}^{+}  \geq \reach{\sG, \{\varocircle, \varotimes\}}^{+}$ are proved similarly using
Lemma~\ref{th:strategy-construction}(2)
and Lemma~\ref{th:strategy-construction}(3), respectively.

\end{proof}


\subsection*{Proof of Theorem \ref{thm:termination}}

\begin{theorem}
Algorithm~\ref{alg:compute-g} terminates, and its result $\sG$ is a valid abstraction.
\end{theorem}
\begin{proof} Assume for the sake of contradiction that Algorithm~~\ref{alg:compute-g} does not terminate. 
  Since every node in $\sG$ has only finitely many successors, the spanning tree of $\sG$ contains an 
  infinite branch $s_0 \xrightarrow {a_0} s_1 \xrightarrow {a_1} s_2 \ldots$ 
  by K\"{o}nig's lemma, and at least one action $a \in \sC$ 
  appears infinitely often in the branch, since $\sC$ is finite. Let $q_0, t_0, q_1, t_1 \ldots$ be the sequence of 
  all nodes in the branch such that $q_l \xrightarrow a t_l$ for all $l \in \mathbb{N}$.
  Then, by the definition of \texttt{EXTRAPOLATE}, there exists a sequence $v_0, v_1, \ldots$ of elements in $D^\sharp$ such that
  $t_{l+1} = t_{{l}} \nabla (v_l \sqcup t_{l}) \sqsupseteq t_{l}$ for all $l \geq 0$, and so
  $t_{0} \sqsubseteq t_{1} \sqsubseteq \ldots$ holds.
  Define $a_0 := t_{0}$ and $a_{l+1} := v_{l} \sqcup t_{l}$ for $l \geq 0$. A simple induction shows that this sequence is monotonically increasing.
  Since $t_1 = a_1$ and for $l \geq 0$ it holds that $t_{l+1} = t_{{l}} \nabla (v_l \sqcup t_{l}) = t_l \nabla a_{l+1}$, we 
  conclude from the definition of a widening operator from Section~\ref{subsec:abstractnpp} that there is  a number $k$ such that 
  $t_{k} = t_{k+1}$, a contradiction to the assumption that the branch is infinite (since then we would have a cycle).

   We already pointed out that every newly constructed node in $\sG$ satisfies the conditions from Def.~\ref{def:validabstraction} and
   $\texttt{EXTRAPOLATE}(a) \sqsupseteq a$ holds for all $a \in D^\sharp$, we can conclude that $\sG$ is a valid abstraction.
\end{proof}